\documentclass[prx,a4paper,aps,twocolumn,superscriptaddress]{revtex4-1}

\usepackage[utf8]{inputenc}
\usepackage[T1]{fontenc}
\usepackage{amsmath, amsthm, amssymb, mathtools, dsfont}
\usepackage{times}
\usepackage{graphicx}
\usepackage{dcolumn}
\usepackage{bm,bbm}
\usepackage{color}
\usepackage[usenames,dvipsnames]{xcolor}
\usepackage{esint}
\usepackage{paralist}

\usepackage{geometry}
\definecolor{myurlcolor}{rgb}{0,0,0.7}
\definecolor{myrefcolor}{rgb}{0.8,0,0}
\usepackage[unicode=true,pdfusetitle, bookmarks=true,bookmarksnumbered=false,bookmarksopen=false, breaklinks=false,pdfborder={0 0 0},backref=false,colorlinks=true, linkcolor=myrefcolor,citecolor=myurlcolor,urlcolor=myurlcolor]{hyperref}

\DeclareGraphicsExtensions{.png,.pdf,.eps}
\graphicspath{ {./figures/} }
\usepackage{epstopdf}

\makeatletter

 \theoremstyle{plain}
 
 \theoremstyle{plain}
 \newtheorem{lem}{Lemma}
 \theoremstyle{plain}
 \newtheorem{thm}{Theorem}
 \theoremstyle{plain}
 \newtheorem{exa}{Example}
 \theoremstyle{plain}
 
 \theoremstyle{plain}
 
 \theoremstyle{remark}
 \newtheorem*{rem*}{Remark}
	\theoremstyle{plain}
	 \newtheorem{rem}{Remark}

\newtheorem{innercustomgeneric}{\customgenericname}
\providecommand{\customgenericname}{}
\newcommand{\newcustomtheorem}[2]{%
  \newenvironment{#1}[1]
  {%
   \renewcommand\customgenericname{#2}%
   \renewcommand\theinnercustomgeneric{##1}%
   \innercustomgeneric
  }
  {\endinnercustomgeneric}
}

\newcustomtheorem{customthm}{Theorem}
\newcustomtheorem{customlemma}{Lemma}

\newcommand{\e}{\mathrm{e}}
\newcommand{\ot}{\otimes}
\newcommand{\ii}{\mathrm{i}} 

\renewcommand{\exp}{\mathrm{exp}}
\newcommand{\sgn}{\mathrm{sgn}}
\newcommand{\DEF}{\coloneqq}
\DeclareMathOperator{\tr}{tr}

\renewcommand{\H}{\mathcal{H}} 
\newcommand{\W}{\mathcal{W}} 
\newcommand{\J}{\mathcal{J}} 

\newcommand{\Hbos}{\mathcal{H}_b}
\newcommand{\Hfer}{\mathcal{H}_f}
\newcommand{\Hfree}{\mathcal{H}^{+}_{\mathrm{Fock}}}
\newcommand{\Hfock}{\mathcal{H}_{\mathrm{Fock}}}

\newcommand{\LOB}{\mathrm{LO}_b} 
\newcommand{\LOF}{\mathrm{LO}_f} 
\newcommand{\FLO}{\mathrm{FLO}} 

\newcommand{\GEN}{\mathcal{G}}

\renewcommand{\U}{\mathrm{U}} 
\newcommand{\SU}{\mathrm{SU}} 
\newcommand{\SO}{\mathrm{SO}} 
\newcommand{\USP}{\mathrm{USp}} 
\newcommand{\Spin}{\mathrm{Spin}} 

\newcommand{\su}{\mathfrak{su}} 
\newcommand{\so}{\mathfrak{so}} 
\newcommand{\usp}{\mathfrak{usp}} 
\newcommand{\g}{\mathfrak{g}} 
\renewcommand{\k}{\mathfrak{k}} 

\newcommand{\defeq}{\coloneqq}

\renewcommand{\C}{\mathbb{C}} 
\newcommand{\R}{\mathbb{R}} 
\newcommand{\T}{\mathbb{T}} 
\newcommand{\M}{\mathbb{M}} 
\newcommand{\Lie}{\mathrm{Lie}} 
\newcommand{\Herm}{\mathrm{Herm}}  
\newcommand{\Aut}{\mathrm{Aut}}  
\newcommand{\Ad}{\mathrm{Ad}}
\newcommand{\Inn}{\mathrm{Inn}}  
\newcommand{\I}{\mathbbmss{1}} 
\renewcommand{\L}{\mathbb{L}} 
\renewcommand{\P}{\mathbb{P}} 

\global\long\global\long\global\long\def\bra#1{\mbox{\ensuremath{\langle#1|}}}
\global\long\global\long\global\long\def\ket#1{\mbox{\ensuremath{|#1\rangle}}}

\global\long\global\long\global\long\def\kb#1#2{\mbox{\ensuremath{\ensuremath{\ensuremath{|#1\rangle\!\langle#2|}}}}}

\global\long\global\long\global\long\def\SET#1#2{\mbox{\ensuremath{\ensuremath{\left\lbrace\left. #1\ \right|\ #2\right\rbrace }}}}

\makeatother

\begin{document}

\title{Universal extensions of restricted classes of quantum operations}

\author{Micha\l\ Oszmaniec}
\affiliation{ICFO-Institut de Ci\`encies Fot\`oniques, The Barcelona Institute of Science and Technology, 08860 Castelldefels (Barcelona), Spain}
\affiliation{National Quantum Information Centre of Gda\'nsk, 81-824 Sopot, Poland}
\affiliation{Faculty of Mathematics, Physics and Informatics, 
	Institute of Theoretical Physics and Astrophysics, University of
	Gda\'nsk, 80-952 Gdañsk, Poland}

\author{Zolt\'an Zimbor\'as}
\affiliation{Wigner Research Centre for Physics, Hungarian Academy of Sciences, P.O. Box 49, H-1525 Budapest, Hungary}
\affiliation{Dahlem Center for Complex Quantum Systems, Freie Universit\"at Berlin, 14195 Berlin, Germany}

\begin{abstract}
For numerous applications of quantum theory it is desirable to be able to apply arbitrary unitary operations on a given quantum system. However, in particular situations only a subset of unitary operations is easily accessible. This raises  the question of what additional unitary gates should be added to a given gate-set in order to attain physical universality, i.e., to be able to perform arbitrary unitary transformation on the relevant Hilbert space. In this work, we study this problem for three paradigmatic cases of naturally occurring restricted gate-sets: (A) particle-number preserving  bosonic linear optics, (B) particle-number preserving fermionic linear optics, and (C) general (not necessarily particle-number preserving) fermionic linear optics. Using tools from group theory and control theory, we classify, in each of these scenarios, what sets of gates are generated, if an additional gate is added to the set of allowed transformations. This allows us to solve the universality problem completely for arbitrary number of particles and for arbitrary dimensions of the single-particle Hilbert space. 
\end{abstract}

\maketitle

In many applications of quantum mechanics it is important to have full control over a quantum system used to perform a desired task or a quantum protocol. This amounts to being able to implement arbitrary unitary operation on the system in question.  Perhaps the most well-known example is the circuit model of quantum computing, where the ability to implement arbitrary unitary gates on a system of many distinguishable particles (say, qubits) is a necessary ingredient for performing universal quantum computation  \cite{Kitaev2002,Nielsen2010}. From the experimental perspective, it is typically very easy to implement single-qubit gates. This collection of gates, however, does not lead to universal quantum computation  and to this aim has to be supplemented by an entangling gate \cite{Brylinski2002}. Similar situations appears in other physical contexts. Typically, the set of \emph{easily accessible}  unitary gates acting on a given quantum system, does not ensure full controllability. 

This work studies the \emph{extension problem} for  gate-sets appearing naturally in systems consisting of non-distinguishable particles: passive linear optics for  (A) bosonic  and  (B) fermionic systems with  fixed number of particles, as well as (C) active linear optics acting on fermionic system with fixed number of modes subject to the parity super-selection rule \cite{Bravyi2002}. Specifically, for the aforementioned scenarios, we study what unitary transformations \emph{on the relevant Hilbert space} can be implemented if the restricted class of gates $K$ is supplemented by additional unitary transformations - see Fig. \ref{fig:protocol}.  We investigate two variants of this problem:

\begin{itemize}
	\item  [(i) ] the gate-set $K$ is supplemented with unitaries of the form $\exp\left(-\ii tX\right)$ generated by the Hamiltonian $X$;
	\item [(ii) ] the gate-set $K$ is supplemented by a \emph{single} unitary transformation $V$.
\end{itemize}

\begin{figure}[!t]
\centering
\includegraphics[width=0.8\columnwidth]{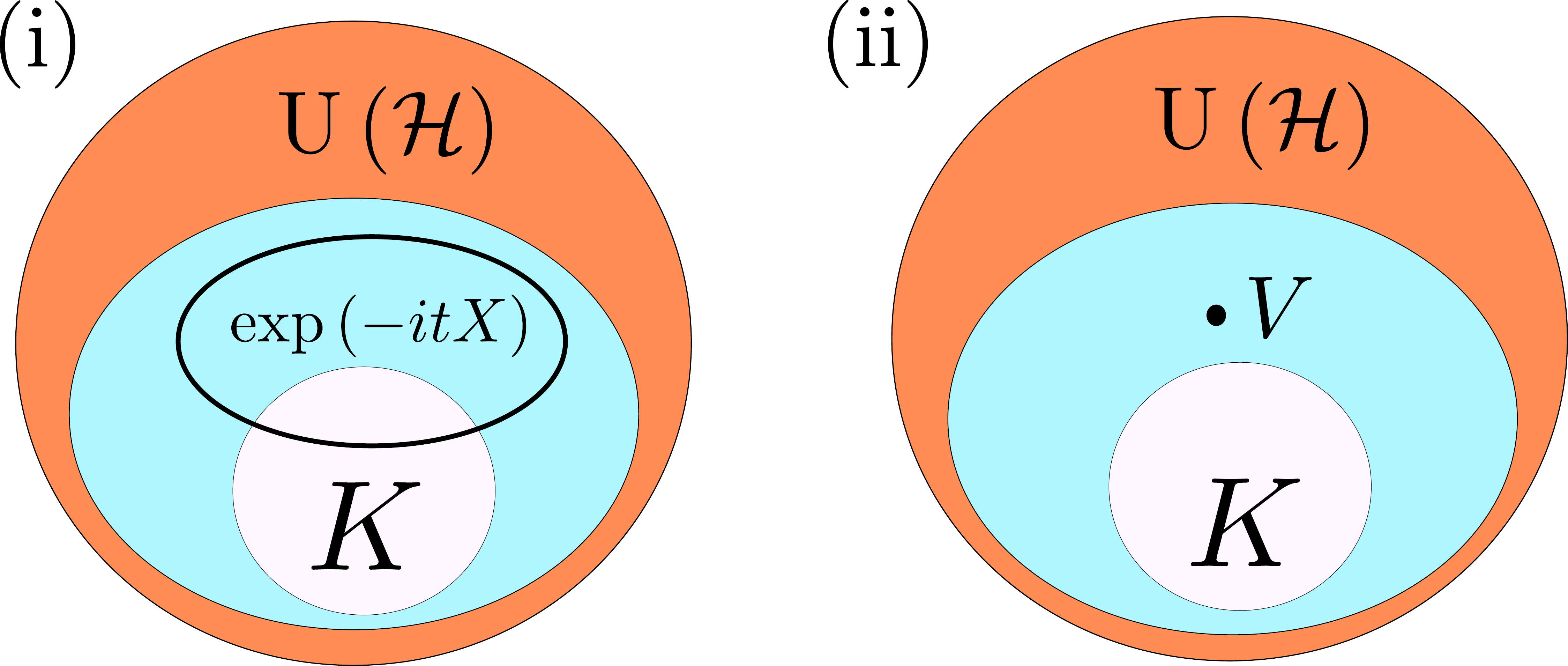}
\caption{A schematic presentation of the problems studied. (i) Given a family of gates $K$ (white) and a one-parameter family of unitaries $\exp(-\ii tX)$ (black loop), what class of gates (cyan) can be generated in the full unitary group $\U\left(\H\right)$ (orange)? (ii) Given a family of gates $K$ (white) and a single gate $V$ (black dot), what class of gates  (cyan) can be generated in the full unitary group $\U\left(\H\right)$ (orange)?}
\label{fig:protocol}
\end{figure}

Linear optical transformations are relevant in many contexts. Passive bosonic linear optics describes single-particle evolutions of a system of $N$ identical bosons in $d$ modes. Such transformations are natural for quantum optics, when quantum states of light pass through an optical network formed from beam-splitters and phase shifters \cite{Carolan2015}. Linear optics underpins the KLM scheme of quantum computing with photons \cite{Knill2001,Kok2007} and the boson sampling strategy for demonstrating quantum supremacy with linear optical networks \cite{Aaronson2011}. Moreover, this class of transformations is used to manipulate cold bosonic particles in optical traps \cite{Milburn1997,Albiez2005}. Similarly, passive fermionic linear optics describes single-particle evolutions of a system of $N$ identical fermions in $d$ modes \cite{Terhal2002,DiVincenzo2005}, which  can be realized in  integer quantum hall effect systems exhibiting edge channels \cite{Bocquillon2014}.  Passive fermionic  linear optics together with particle-number measurements yields classically simulable model of quantum computation \cite{DiVincenzo2005}. Moreover, this class of transformations has been recently used to study correlation \cite{Tichy2014} and nonlocality \cite{Dasenbrook2016} properties in fermionic systems. Finally, active fermionic linear optics describes free-fermion transformations that are not necessary particle preserving.  These fermionic transformations are the basic ingredient of a classically simulable model of quantum computation \cite{Bravyi2002,Terhal2002,DiVincenzo2005}. This computational model  has been even explored in the presence of noise \cite{Bravyi2011,Melo2013,Oszmaniec2014}, and can be connected,  via the Jordan-Wigner transformation, to the model of computation based on Matchgates \cite{Valiant2002,Knill2001a,Jozsa2008,Brod2016}.

In this work, we completely solve problems (i) and (ii) for the scenarios A-C.  We characterize the unitary transformations that are implementable (maybe approximately) by linear optical gates supplemented with any additional  Hamiltonian or a gate. Our characterization is given in terms of \emph{explicit} algebraic conditions on the Hamiltonian $X$ or the gate $V$ that can be can be tested operationally. The resulting behavior is surprisingly rich and structurally depends on the number of modes and the number of particles. In particular, contrary to what intuition might suggest, it is not true that every non-trivial extra gate or Hamiltonian provides universality in scenarios A-C. Solution of problems (i) and (ii) gives the clear understanding of what resources are necessary to have full physical controllability in the contexts listed above. Moreover, our results can be viewed as a step towards a solution of the general problem of classification of invertible quantum circuits posed recently by Aaronson and co-workers \cite{Bouland2014,Aaronson2015}. On the technical level we use techniques of  the theory of Lie groups and Lie algebras, which  have recently proved useful in studies on controllability of quantum systems \cite{Altafini2002, Zeier2011,ZZKS2014, ZZSB2015} and on universality of gate-sets \cite{Sawicki2016,GatesSawickiLong2016,GatesSawickiLong2016}. 

Our general results have application  to concrete physical examples. First, we consider the problem of extension to universality of passive bosonic linear optics for $d=2$ modes via cross-Kerr interaction \cite{Brod2016a}, which is of relevance in  quantum metrology with random bosonic states \cite{Oszmaniec2016}. We also show that, quite surprisingly, there exist simple nonlinear Hamiltonians that do not lead to universality when added to passive fermionic or bosonic linear optics. Finally, it turns out that a simple quartic (in Majorana operators) Hamiltonian promotes active fermionic linear optics to universality in the positive-parity subspace. 
  
\paragraph*{Setting---}
Let us start with formally defining the extension  problem  for  (i) an additional Hamiltonian and (ii) an additional gate. In general we say that a unitary gate $U$ can be generated by unitaries from a set  $\mathcal{S}$  if  it can be approximated with  \emph{arbitrary precision} (in operator norm) by a sequence of products of gates from $\mathcal{S}$. We denote by $\langle K, X  \rangle$  the set of unitaries that can be generated  form the restricted gate-set $K$ and unitaries of the form $\exp\left(-\ii tX\right)$, where $t$ is an arbitrary real number. Likewise, slightly abusing the notation, we denote by $\langle K,V \rangle$ the set of unitaries that can be generated by elements form $K$ and an extra gate $V$.  The aim of this work is to characterize sets  $\langle K, X  \rangle$  and $\langle K,V  \rangle$ for different linear optical groups $K$ (acting on the appropriate Hilbert spaces $\H$). If the resulting gate-set $\langle K, X  \rangle$  (or  $\langle K,V  \rangle$) form the full unitary group  $\U(\H)$,  we say that the Hamiltonian $X$ (or the gate $V$)  \emph{promotes} the restricted collection of gates $K$ to universality in $\H$. Since unitaries $U$ and $e^{i\alpha} U$ are physically indistinguishable, we assume, without loss of generality, that gates of the form $\exp(\ii \theta) \I$ are contained in $K$. In what follows by $\T(\H)$ we will denote the unitaries proportional to $\I$ on $\H$. We would like to remark that the notion of \emph{physical universality} in $\H$ introduced above does not imply computational universality in a sense of complexity theory \cite{Nielsen2010}.   

We denote the Hilbert space of $N$ bosons in $d$ modes by $\Hbos$. We have $\Hbos= \mathrm{Sym}^N\left(\mathbb{C}^d\right)$, i.e., in this case the bosonic Hilbert space can be identified with  the totally symmetric subspace of the Hilbert space of $N$ distinguishable qudits, $\left(\mathbb{C}^d\right)^{\ot N}$. In this language the group of passive linear optical bosonic transformations, denoted by $\LOB$, can defined as the group of unitaries of the form $U^{\ot N}$, with $U\in \U(d)$, restricted to the bosonic subspace $\Hbos$. 

The Hilbert space of $N$ (spinless) fermionic particles in $d$ modes ($d\geq N$) is $\Hfer=\wedge^N\left(\C^d\right)$, i.e., the totally antisymmetric subspace of  $\left(\mathbb{C}^d\right)^{\ot N}$.  Similarly to the bosonic case, the group of passive fermionic linear optics $\LOF$  is defined as the group of unitaries $U^{\ot N}$, with $U\in \U(d)$, restricted to the fermionic subspace $\Hfer$.

In the case of fermionic system without the restricted number of particles, the relevant Hilbert space is direct sum of the different $N$-particle fermionic Hilbert spaces, i.e. the Fock space $\mathcal{H}_{\mathrm{Fock}}{=}\oplus_{N=0}^d \wedge^N\C^d$.  $\mathcal{H}_{\mathrm{Fock}}$ is spanned by the Fock states $\ket{n_1,\ldots,n_d}$, with $n_k \in \lbrace0,1 \rbrace$. In this space we have fermionic creation and annihilation operators, $f_{k}^{\dagger}$, $f_{k}$ satisfying the canonical anti-commutation relations and one can define the standard number operators $\hat{n}_k = f^\dagger_k f^{\phantom\dagger}_k$. It is also convenient to introduce Majorana fermion operators $m_{2k-1}{=} f^{\phantom\dagger}_{k}{+}f_{k}^{\dagger}$, $m_{2k}{=}i\left(f^{\phantom\dagger}_{k}{-}f_{k}^{\dagger}\right)$,
for $k{=}1,{\ldots},d$. Majorana operators satisfy the anticommutation relations $\left\{ m_{i},\,
m_{j}\right\} =2\delta_{ij} \I$. 
In many situations fermionic systems obey the so-called \emph{parity superselection rule} \cite{Bravyi2002,ZZKS2014,Bravyi2004} which states that any physically accessible operations must commute with the total parity operator $Q=(-1)^{\sum_{k=1}^d \hat{n}_k}$. In this work, we restrict our attention to the positive-parity subspace $\Hfree$, i.e., the subspace spanned by Fock states with even number of particles. 
Results completely analogous to the ones presented here can be obtained  also for the negative-parity subspace.  Active fermionic linear optics acting in $\mathcal{H}_{\mathrm{Fock}}$ consists of unitaries of the form $\exp(\sum_{l, k=1}^{2d}h_{kl} m_k m_l)$, where $h_{kl}$ is a real antisymmetric $2d\times2d$ matrix. Since  we are interested only in the positive-parity subspace $\Hfree$, we formally define the group of active fermionic linear transformations as the group consisting of unitaries $(\exp(\ii \phi) \I) \exp(\sum_{ l, k =1}^{2d}h_{kl} m_k m_l)$, restricted to $\Hfree$. 

The groups introduced above are compact Lie subgroups of the unitary group $\U\left(\H\right)$, where $\H$ is the Hilbert space describing the relevant physical system.  For a given Lie group $K\subset \mathrm{U}\left(\H\right)$ it is convenient to work with the corresponding Lie algebra consisting of Hamiltonians that generate (via exponentiation) unitaries belonging to $K$.  The groups studied in this paper and the corresponding  Lie algebras are closely related to the irreducible representations of \emph{simple} Lie algebras \cite{HallGroups} in the relevant Hilbert spaces. This observation allowing us to obtain our central results given in Theorems 1-4 below. In particular, we use the classification results by Dynkin \cite{Dynkin2000} concerning the maximal Lie subalgebras of simple Lie algebras.  For the sake of clarity of the presentation, we moved technical proofs to the Appendix and focused on the discussion of the physical meaning of our results.  
 
\paragraph*{Application---}

Before stating our results in full generality, let us present first an exemplary application of our findings. In the reference \cite{Oszmaniec2016}, the  authors were interested in extending  bosonic linear optics to universality in $\Hbos$ by adding an additional gate. This problem was motivated by the need to construct physically-accessible  universal gate-set in $\Hbos$, which can be used to generate, via construction based on random circuits \cite{Horo2016}, approximate bosonic $t$-designs. The example below proves that a singe gate based on the cross-Kerr nonlinearity suffices to promote bosonic linear optics to universality in $\Hbos$. It should be mentioned that Kerr-like transformation have been previously used to obtain universal quantum computation in continuous-variable systems \cite{lloydcont1999}.

\begin{exa}\label{ex:CROSSKERR}
	Consider a bosonic system with $d=2$ modes and $N>1$ particles, and a gate generated by the cross-Kerr interaction (acting on $\Hbos$ for a fixed time $t$),	
	\begin{equation}
	V_{t}=\exp\left(-it\hat{n}_{a}\hat{n}_{b}\right)\,,\label{eq:cross Kerr}
	\end{equation}
	where $\hat{n}_{a,b}$ are the occupation number operators corresponding
	to modes $a$ and $b$.  Let $\langle \LOB, V_{t}\rangle$ be the group of transformations generated by passive bosonic  linear optics and $V_t$. Then, 	$\langle \LOB, V_{t}\rangle =\U\left(\Hbos \right)$ if and only if 
	\begin{equation}\label{eq:conditionKerr}
	\e^{2it\left[l\left(N-l\right)-k\left(N-k\right)\right]}\neq 1\,,
	\end{equation}
	for at least one pair $(k,l)$ , where $k,l=0,\ldots,N$. In particular, the gate $V_{\frac{\pi}{3}}$ promotes passive bosonic linear optics to universality in $\Hbos\ $ for $d=2$ modes.

	\end{exa}

The above result follows from Theorem \ref{thm:BOSGATES} stated below. The explicit computations leading to condition \eqref{eq:conditionKerr} are presented in Appendix \ref{ap:examples}.

\paragraph*{Main results---} 
We start with the presentation of results concerning passive bosonic linear optics.

\begin{thm}[Extensions of Passive Bosonic  Linear Optics with an additional gate]\label{thm:BOSGATES}
Let $V\notin \LOB$ be a gate acting on the Hilbert space $\Hbos$ of $N>1$ bosons in $d$ modes. Let $\langle \LOB, V\rangle$ be the group of transformations generated by passive bosonic linear optics and $V$ in $\Hbos$. For $d=2$ we define:
\begin{equation}\label{eq:Lbos}
\L _b=\kb{\Psi_b}{\Psi_b},\ \ket{\Psi_b}=\sum_{k=0}^N (-1)^k  \ket{D_k}\ket{D_{N-k}} \in \Hbos\ot\Hbos\ ,
\end{equation}
where $\ket{D_k}$ denote the two-mode Dicke states with $k$-particles being in the first mode.  We have the following possibilities:
\begin{itemize}
\item [(a)] If $d>2$, then $\langle \LOB,V\rangle =\U\left(\Hbos \right)$.
\item [(b)] If $d=2$, $N\neq 6$ and $[V\ot V, \L_b]=0$, then 
\[
\langle \LOB,V\rangle=G_b= \lbrace \left. U\in \U\left(\Hbos \right) \right|    \left[U\ot U, \L_b\right]=0\ \rbrace .
\]

\item [(c)] If $d=2$ and $[V\ot V, \L_b ]\neq 0$, then  $\langle \LOB,V\rangle=\U\left(\Hbos\right)$.
\end{itemize}
\end{thm}
In the above theorem we have situations with $N=1$ particles as for them $\LOB=\U(\Hbos)$. We see that for  $d\neq2$ any additional gate promotes $\LOB$ to universality in the bosonic space $\Hbos$. For $d=2$ the resulting gate-set $\langle \LOB, V\rangle$ depends only on the commutator $[V\ot V,\L_b]$. If it is nonzero, then $V$ again extends $\LOB$ to universality; while if it vanishes (and $N\neq6$) $V$ extends $\LOB$ to the ''middle group'' $G_b$.  Up to a global phase the group $G_b$ consists of unitaries that preserve the bilinear form defined by $B(\ket{\psi},\ket{\phi})=\bra{\Psi_b} (\ket{\psi}\ot\ket{\phi})$. Here, by preservation we mean that $B_b(U\ket{\psi},U\ket{\phi})=B_b(\ket{\psi},\ket{\phi})$, for all vectors $\ket{\phi}$,$\ket{\psi}$. If the number of particles $N$ is even  then $\ket{\Psi_b}$ is a symmetric tensor and defines the real inner product. In this case we have $G_b= \langle \T(\Hbos), \SO(\Hbos) \rangle$, where $SO(\Hbos)$ is the special orthogonal group on $\Hbos$. When the number of particles $N$ is odd, the vector $\ket{\Psi_b}$ is antisymmetric and defines the symplectic structure (i.e. non-degenerate and antisymmetric  form) on $\Hbos$. In this case we have  $G_b = \langle \T(\Hbos), \USP(\Hbos) \rangle$, where 
$\USP(\Hbos)$ is the unitary symplectic group. The two differ considerably as $\USP(\H)$ acts transitively on the set of pure states on $\H$ \cite{ZZKS2014, schirmer2002, albertini2003}. On the other hand, $\SO(\H)$ acts transitively only on ''real'' pure states. Thus, for odd number of particles, adding any additional gate gives either the full unitary controllability or the  pure-state controllability. In the case of $d=2$ modes and $N=6$ particles if $[V\ot V,\L_b]=0$ the situation complicates due to the presence of additional group (related to the \emph{exotic} group $\mathrm{G}_2$) in between $\LOB$ and $G_b$. We leave the description of this exceptional case as an interesting open problem. 

From Theorem \ref{thm:BOSGATES} one can infer the following result concerning the additional Hamiltonian $X$ acting on $\Hbos$  (intuitively,  one can obtain it by setting $V=\exp(\ii t X)$ in Theorem \ref{thm:BOSGATES} and differentiating over $t$).

\begin{thm}[Extensions of Passive Bosonic  Linear Optics via an additional Hamiltonian]\label{thm:BOSHAM}
	Let $X\notin \Lie\left(\LOB\right)$ be a Hamiltonian acting on Hilbert space of $N$ bosons in $d$ modes $\Hbos$. Let $\langle \LOB, X\rangle$ be the group of transformations generated by passive bosonic linear optical optics and $X$ in $\Hbos$. 
	We have the following possibilities:
	\begin{itemize}
		\item[(a)] If $d>2$, then $\langle \LOB, X\rangle=\U\left(\Hbos \right)$.
		\item[(b)] If $d=2$, $N\neq 6$, and $\left[X\ot \I + \I\ot X, \L_b\right]=0$, then 
		\begin{equation}
		\langle \LOB, X\rangle=G_b=\lbrace \left. U\in \U\left(\Hbos \right) \right|   \left[U\ot U, \L_b\right]=0\ \rbrace \ .
		\end{equation}
		\item[(c)] If $d=2$ and $\left[X\ot \I + \I\ot X, \L_b\right]\neq0$, then  $\langle \LOB, X\rangle=\U\left(\Hbos\right)$.
		
	\end{itemize}
		
\end{thm}
Combining the above result and the discussion below Theorem~\ref{thm:BOSGATES} we see that there might exist physical Hamiltonians that add different controllability properties to $\LOB$, depending on the number of particles $N$. The following example shows that this is indeed the case.

\begin{exa}\label{ex:BOSmid} Consider the Hamiltonian $X_{3}=\hat{n}^3_a - \hat{n}^3_b$ acting on $\Hbos$ for $d=2$ modes and $N$ particles.  Deepening on the value $N$ we get different types of gate-sets after supplementing passive bosonic optics with $X_3$ 
\begin{itemize}
\item[(a)] For even $N$: $\langle \LOB, X_3\rangle = \langle \T(\Hbos), \SO(\Hbos) \rangle$; 
\item[(b)] For odd $N$:  $\langle \LOB, X_3\rangle = \langle \T(\Hbos), \USP(\Hbos) \rangle$.

\end{itemize}
In, particular for odd $N$ we have full controllability on the set of pure states on $\H_b$, whereas for even $N$ this is not the case.
\end{exa}
We now move to the discussion of fermionic linear optics (both passive and active). In the main text we present the results concertinaing the gate extension problem (ii). The corresponding theorems for the Hamiltonian extension problem are given in the Appendix~\ref{ap:mainTheorems} and their relation with the results presented in the main text is analogous to the connection between Theorem \ref{thm:BOSHAM} and Theorem \ref{thm:BOSGATES}.

\begin{thm}[Extensions of Passive  Fermionic Linear Optics with an additional gate]\label{thm:FERMGATES}
Let $V\notin \LOF$ be a gate acting on Hilbert space of $N$ fermions in $d$ modes $\Hfer$, where $N\notin\lbrace 0,1,d-1,d\rbrace$. Let $\langle \LOF, V\rangle$ be the group of transformations generated by passive fermionic linear optics and $V$ in $\Hfer$. For $d=2N$ (half-filling) we define:
\begin{equation}
\L_f=\kb{\Psi_f}{\Psi_f},  \text{ with } \ket{\Psi_f}=\ket{1}\wedge\ket{2}\wedge\ldots\wedge\ket{2N} \in \Hfer\ot\Hfer
\label{eq:Lfer}
\end{equation}
where $\wedge$ denotes the standard wedge product.
We have the following possibilities:
\begin{itemize}
\item[(a)] If $d\neq 2N$, then $\langle \LOF, V\rangle=\U\left(\Hfer\right)$.
\item[(b)] If $d=2N$ and  $V= W k$, for $k\in\LOF$ and  $W=\prod_{i=1}^d (f_i + f^\dagger_i)$, then $\langle \LOF, V\rangle = \LOF \cup \LOF \cdot W$.
\item[(c)] If $d=2N$, $V\neq gW $, for $g\in\LOF$, and $[V\ot V, \L_f]=0$, then   \vspace*{-2mm}
\[
\langle \LOF, V\rangle = G_f=\lbrace \left. U\in \U\left(\Hfer \right) \right|    \ [V\ot V, \L_f]=0\ \rbrace.
\] \vspace*{-5mm}
\item[(d)] If $d=2N$ and $[V\ot V, \L_f]\neq 0$, then  $\langle \LOF, V\rangle=\U\left(\Hfer\right)$.
\end{itemize}
\end{thm}

The structure of the above result is similar to the case of passive bosonic linear optics. In the formulation of the theorem we have excluded the non interesting cases $N\in\lbrace 0,1,d-1,d\rbrace$ since for them $\LOF$ equals the full unitary group on the respective Hilbert space. When $d\neq2N$ every gate promotes $\LOF$ to universality. However, in the physically relevant case  of half-filling \cite{Greiter1991}, a more interesting ''onion'' structure  appears. In the case (b) addition of an extra gate of the form $kW$, where $k\in\LOF$ and gate $W$ (describing particle-hole transformation in $\Hfer$) gives  the gate-set $\LOF \cup \LOF \cdot W$ (it is crucial here that $W$ commutes with $\L_f$  and that  conjugation by $W$ leaves $\LOF$ invariant). The further possibilities are described, similarly to the bosonic case, by  the commutation properties of $V\ot V$ with $\LOF$. If  $d$ is not divisible by four we have $G_f=\langle \T(\Hfer), \SO(\Hfer) \rangle$. On the other hand, for $d$ divisible by four,  $G_f=\langle \T(\Hfer), \USP(\Hfer) \rangle$. The corresponding bilinear forms are defined by inner products with $\ket{\Psi_f}$.  Using Theorem~\ref{thm:FERMGATES}, one can  efficiently obtain results on  extensions to universality, as the one given by  the next example.

\begin{exa} For any non-quadratic Hamiltonian $M$ containing only two-mode terms, the generated group $\langle \LOF, M\rangle$ is the entire unitary group $\U\left(\Hfer\right)$.
\end{exa}

\begin{figure}[!t]
\centering
\includegraphics[width=0.5\columnwidth]{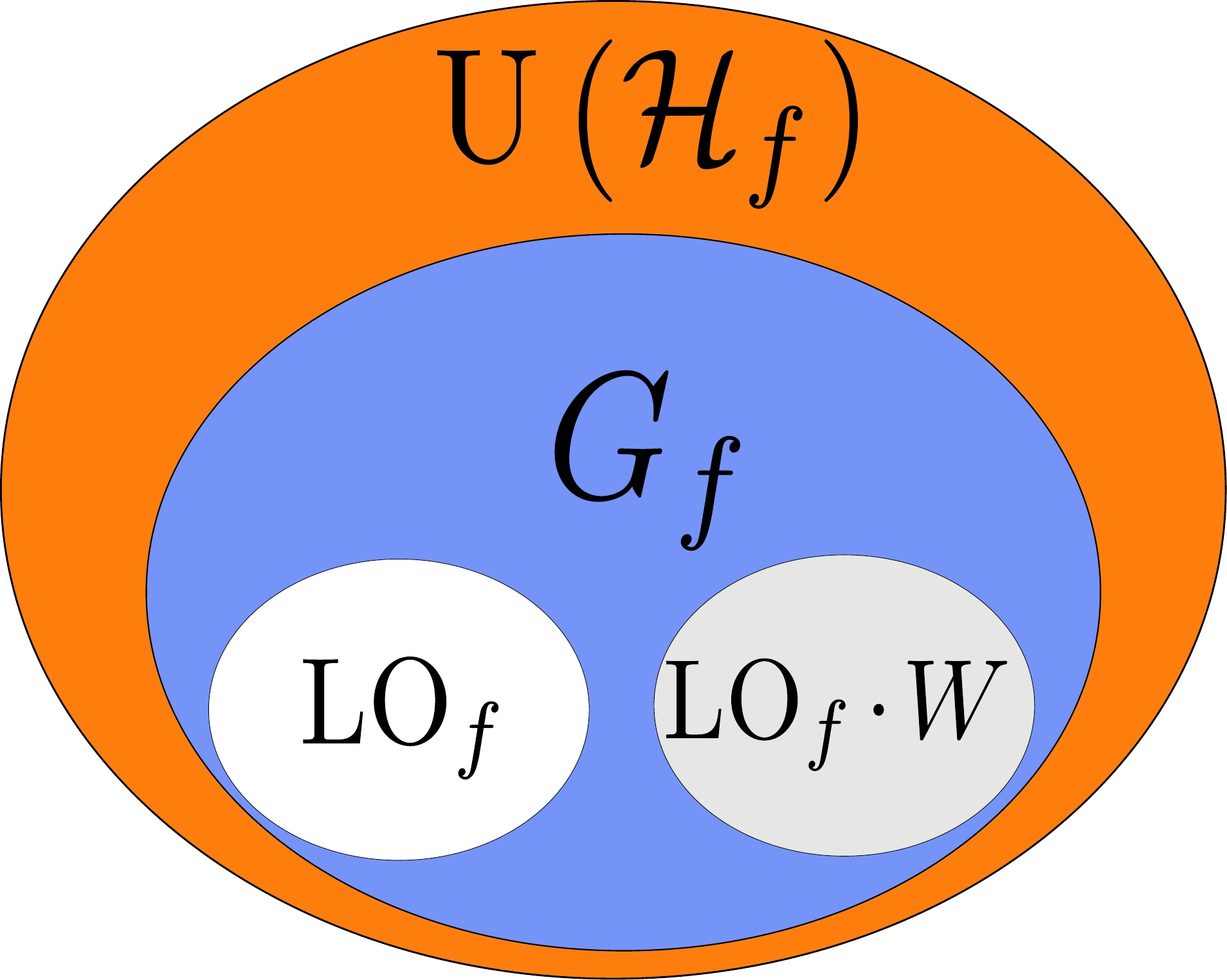}
\caption{
A pictorial presentation of most complicated chain of group inclusions that can be realized for the problem (ii) in the considered scenarios. For passive fermionic linear optics $\LOF$ in the half-filling case ($d=2N$) we have $ \LOF\subset \LOF\cup \LOF\cdot W\subset G_f \subset \U\left(\Hfer\right)$, where $W=\prod_{i=1}^d (f_i + f^\dagger_i)$, and the ''middle'' group $G_f$ is defined by the condition $\left[U\ot U, \L_f \right]=0$. }
\label{fig:result}
\end{figure}

Hamiltonians that are not composed of two-mode terms are also  often studied. One typical family of these are the so-called correlated hopping Hamiltonians, where the hopping-term between two sites is multiplied  with number operators belonging to other sites. For such Hamiltonians universality is not guaranteed:
 
 \begin{exa}Consider the correlated hopping Hamiltonian 
 \begin{align}
 Y=\sum_{j=1}^{d/2-1}& (\hat{n}_{2j}-\hat{n}_{2j+2})^2(f_{2j-1}^{\dagger}f_{2j+1}^{\phantom\dagger}+ f_{2j+1}^{\dagger}f_{2j-1}^{\phantom\dagger})\, + \nonumber
 \\&(\hat{n}_{2j-1}{-}\hat{n}_{2j+1})^2(f_{2j}^{\dagger} f_{2j+2}^{\phantom\dagger} + f_{2j+2}^{\dagger}f_{2j}^{\phantom\dagger}). \label{eq:corr_hopp}
 \end{align}
acting on $\Hfer$ for the case of half filling ($d=2N$).  Then, we have the following situations 
\begin{itemize}
\item[(a)] for even $N$: $\langle \LOF, Y \rangle = \langle \T(\Hfer), \SO(\Hfer) \rangle$; 
\item[(b)] for odd $N$:  $\langle \LOF, Y \rangle = \langle \T(\Hfer), \USP(\Hfer) \rangle$.
\end{itemize}
For odd $N$ the Hamiltonian $Y$ together with $\LOF$ ensures full controllability on the set of pure states on $\Hfer$. However, for even $N$ this is not the case. The above statements are even true for each term appearing in sum Eq.~\eqref{eq:corr_hopp}. The correlated hopping Hamiltonian $Y$ often appears (in a relabeled form) in the literature on extended Hubbard models \cite{DM2013}.
\end{exa}

Our last result concerns the extension problem for active fermionic linear optics.
\begin{thm}[Extensions of active fermionic linear optics via additional gate]\label{thm:FLOGATES}
. Let $\langle \FLO, V\rangle$ be the group of transformations generated by active linear optics and $V$ acting in positive-parity Fock subspace $\Hfree$ with $d>3$ modes. For $d=2k$ we define:
\begin{equation}\label{eq:LFLO}
\L_\FLO{=} \frac{1}{2^{d(2d-1)}}  \prod_{1\leq i<j\leq 2d} \left(I\ot I + m_i m_j\ot m_i m_j \right)  .
\end{equation}
We have the following possibilities:
\begin{itemize}
\item[(a)] If $d\neq 2k$, then $\langle \FLO, V\rangle=\U\left(\Hfree\right)$ 
\item[(b)] If $d=2k$, and $[V\ot V, \L_{\FLO}]=0$, then 
\[
\langle \FLO, V\rangle = G_{\FLO}=\lbrace \left. U\in \U\left(\Hfree \right) \right|    [U\ot U, \L_{\FLO}]=0\ \rbrace.
\]
\item[(c)] If $d=2k$ and $[V\ot V, \L_{\FLO}]\neq 0$, then  $\langle \FLO, V\rangle=\U\left(\Hfree\right)$.
\end{itemize}
\end{thm}

In the above result be have omitted cases $d\leq 3$ as for them $\FLO$ is itself the full unitary group on $\Hfree$ \cite{Melo2013}. If the number of modes is not even, than every gate promotes $\FLO$ to universality in $\Hfree$. If the number of modes is even, then an additional gate either (c) promotes $\FLO$ to universality or extends it to the group $G_\FLO$ (b). Furthermore, analogously to the case of passive fermionic optics, if $d$ is divisible by four $G_\FLO=\langle \T(\Hfree), \SO(\Hfree) \rangle$ and if this is not the case $G_\FLO=\langle \T(\Hfree), \USP(\Hfree) \rangle$ (a detailed description of the corresponding orthogonal and symplectic forms is given in Appendix \ref{ap:aux}).

\begin{exa}\label{ex:FLO}For arbitrary number of modes $d>3$ the interaction $H_{\mathrm{in}}=m_1 m_2 m_3 m_3$ extends $\FLO$ to universality in $\Hfree$ i.e. $\langle \FLO , H_{\mathrm{in}} \rangle = \U(\Hfree)$. This result suggest that time-independent hamiltonain $\H_{\mathrm{in}}$  together with linear optics allows to perform efficient quantum computation (if the standard occupation-number measurements are allowed) \cite{Brod2014,Bravyi2016} .
\end{exa}

\paragraph*{Discussion---}
In this letter, we presented a comprehensive treatment of the extension problems for various classes of linear optical gates for bosons and fermions. The resulting behavior is surprisingly rich and critically depends on the number of modes and number of particles present in the system.  However, there is a number of interesting problems we did not addressed here. First, it would be interesting to analyze which extra gates or Hamiltonians allow for the most efficient control \cite{Nielsen2006} or the  efficient approximation of gates from the appropriate unitary group \cite{Harrow2002}. Another important problem concerns the robustness of the extra gate or Hamiltonian to the noise that inevitably affects any quantum system.  In future works we also plan to use our results to study (computational) universality  of classically simulable models of computation supported on fermionic systems \cite{Terhal2002} and Machgates \cite{Valiant2002,Knill2001a,Jozsa2008,Brod2016}.

\begin{acknowledgments} 
We thank Antonio Ac\'in, Jens Eisert, Christian Gogolin, Adam Sawicki, and Robert Zeier for interesting and stimulating discussions. This work has been supported by the European Research Council (Consolidator Grant 683107/CoG QITBOX), Spanish MINECO (QIBEQI FIS2016-80773-P,  and Severo Ochoa Grant No. SEV-2015-0522), Fundaci'o Privada Cellex, Generalitat de Catalunya (Grant No. SGR 874, 875, and CERCA Programme). M.O acknowledges the support of Homing programme of the Foundation for Polish Science co-financed by the European Union under the European Regional Development Fund.

\end{acknowledgments} 

\bibliography{univref}

\begin{thebibliography}{55}%
\makeatletter
\providecommand \@ifxundefined [1]{%
 \@ifx{#1\undefined}
}%
\providecommand \@ifnum [1]{%
 \ifnum #1\expandafter \@firstoftwo
 \else \expandafter \@secondoftwo
 \fi
}%
\providecommand \@ifx [1]{%
 \ifx #1\expandafter \@firstoftwo
 \else \expandafter \@secondoftwo
 \fi
}%
\providecommand \natexlab [1]{#1}%
\providecommand \enquote  [1]{``#1''}%
\providecommand \bibnamefont  [1]{#1}%
\providecommand \bibfnamefont [1]{#1}%
\providecommand \citenamefont [1]{#1}%
\providecommand \href@noop [0]{\@secondoftwo}%
\providecommand \href [0]{\begingroup \@sanitize@url \@href}%
\providecommand \@href[1]{\@@startlink{#1}\@@href}%
\providecommand \@@href[1]{\endgroup#1\@@endlink}%
\providecommand \@sanitize@url [0]{\catcode `\\12\catcode `\$12\catcode
  `\&12\catcode `\#12\catcode `\^12\catcode `\_12\catcode `\%12\relax}%
\providecommand \@@startlink[1]{}%
\providecommand \@@endlink[0]{}%
\providecommand \url  [0]{\begingroup\@sanitize@url \@url }%
\providecommand \@url [1]{\endgroup\@href {#1}{\urlprefix }}%
\providecommand \urlprefix  [0]{URL }%
\providecommand \Eprint [0]{\href }%
\providecommand \doibase [0]{http://dx.doi.org/}%
\providecommand \selectlanguage [0]{\@gobble}%
\providecommand \bibinfo  [0]{\@secondoftwo}%
\providecommand \bibfield  [0]{\@secondoftwo}%
\providecommand \translation [1]{[#1]}%
\providecommand \BibitemOpen [0]{}%
\providecommand \bibitemStop [0]{}%
\providecommand \bibitemNoStop [0]{.\EOS\space}%
\providecommand \EOS [0]{\spacefactor3000\relax}%
\providecommand \BibitemShut  [1]{\csname bibitem#1\endcsname}%
\let\auto@bib@innerbib\@empty
\bibitem [{\citenamefont {Kitaev}\ \emph {et~al.}(2002)\citenamefont {Kitaev},
  \citenamefont {Shen},\ and\ \citenamefont {Vyalyi}}]{Kitaev2002}%
  \BibitemOpen
  \bibfield  {author} {\bibinfo {author} {\bibfnamefont {A.~Y.}\ \bibnamefont
  {Kitaev}}, \bibinfo {author} {\bibfnamefont {A.}~\bibnamefont {Shen}}, \ and\
  \bibinfo {author} {\bibfnamefont {M.~N.}\ \bibnamefont {Vyalyi}},\
  }\href@noop {} {\emph {\bibinfo {title} {Classical and quantum
  computation}}},\ Vol.~\bibinfo {volume} {47}\ (\bibinfo  {publisher}
  {American Mathematical Society Providence},\ \bibinfo {year}
  {2002})\BibitemShut {NoStop}%
\bibitem [{\citenamefont {Nielsen}\ and\ \citenamefont
  {Chuang}(2010)}]{Nielsen2010}%
  \BibitemOpen
  \bibfield  {author} {\bibinfo {author} {\bibfnamefont {M.~A.}\ \bibnamefont
  {Nielsen}}\ and\ \bibinfo {author} {\bibfnamefont {I.~L.}\ \bibnamefont
  {Chuang}},\ }\href@noop {} {\emph {\bibinfo {title} {Quantum computation and
  quantum information}}}\ (\bibinfo  {publisher} {Cambridge university press},\
  \bibinfo {year} {2010})\BibitemShut {NoStop}%
\bibitem [{\citenamefont {Brylinski}\ and\ \citenamefont
  {Brylinski}(2002)}]{Brylinski2002}%
  \BibitemOpen
  \bibfield  {author} {\bibinfo {author} {\bibfnamefont {J.-L.}\ \bibnamefont
  {Brylinski}}\ and\ \bibinfo {author} {\bibfnamefont {R.}~\bibnamefont
  {Brylinski}},\ }\href@noop {} {\bibfield  {journal} {\bibinfo  {journal}
  {Mathematics of Quantum Computation}\ }\textbf {\bibinfo {volume} {79}}
  (\bibinfo {year} {2002})}\BibitemShut {NoStop}%
\bibitem [{\citenamefont {Bravyi}\ and\ \citenamefont
  {Kitaev}(2002)}]{Bravyi2002}%
  \BibitemOpen
  \bibfield  {author} {\bibinfo {author} {\bibfnamefont {S.~B.}\ \bibnamefont
  {Bravyi}}\ and\ \bibinfo {author} {\bibfnamefont {A.~Y.}\ \bibnamefont
  {Kitaev}},\ }\href {https://doi.org/10.1006/aphy.2002.6254} {\bibfield
  {journal} {\bibinfo  {journal} {Annals of Physics}\ }\textbf {\bibinfo
  {volume} {298}},\ \bibinfo {pages} {210} (\bibinfo {year}
  {2002})}\BibitemShut {NoStop}%
\bibitem [{\citenamefont {Carolan}\ \emph {et~al.}(2015)\citenamefont
  {Carolan}, \citenamefont {Harrold}, \citenamefont {Sparrow}, \citenamefont
  {Mart{\'\i}n-L{\'o}pez}, \citenamefont {Russell}, \citenamefont
  {Silverstone}, \citenamefont {Shadbolt}, \citenamefont {Matsuda},
  \citenamefont {Oguma}, \citenamefont {Itoh} \emph {et~al.}}]{Carolan2015}%
  \BibitemOpen
  \bibfield  {author} {\bibinfo {author} {\bibfnamefont {J.}~\bibnamefont
  {Carolan}}, \bibinfo {author} {\bibfnamefont {C.}~\bibnamefont {Harrold}},
  \bibinfo {author} {\bibfnamefont {C.}~\bibnamefont {Sparrow}}, \bibinfo
  {author} {\bibfnamefont {E.}~\bibnamefont {Mart{\'\i}n-L{\'o}pez}}, \bibinfo
  {author} {\bibfnamefont {N.~J.}\ \bibnamefont {Russell}}, \bibinfo {author}
  {\bibfnamefont {J.~W.}\ \bibnamefont {Silverstone}}, \bibinfo {author}
  {\bibfnamefont {P.~J.}\ \bibnamefont {Shadbolt}}, \bibinfo {author}
  {\bibfnamefont {N.}~\bibnamefont {Matsuda}}, \bibinfo {author} {\bibfnamefont
  {M.}~\bibnamefont {Oguma}}, \bibinfo {author} {\bibfnamefont
  {M.}~\bibnamefont {Itoh}},  \emph {et~al.},\ }\href
  {https://doi.org/10.1126/science.aab3642} {\bibfield  {journal} {\bibinfo
  {journal} {Science}\ }\textbf {\bibinfo {volume} {349}},\ \bibinfo {pages}
  {711} (\bibinfo {year} {2015})}\BibitemShut {NoStop}%
\bibitem [{\citenamefont {Knill}\ \emph {et~al.}(2001)\citenamefont {Knill},
  \citenamefont {Laflamme},\ and\ \citenamefont {Milburn}}]{Knill2001}%
  \BibitemOpen
  \bibfield  {author} {\bibinfo {author} {\bibfnamefont {E.}~\bibnamefont
  {Knill}}, \bibinfo {author} {\bibfnamefont {R.}~\bibnamefont {Laflamme}}, \
  and\ \bibinfo {author} {\bibfnamefont {G.~J.}\ \bibnamefont {Milburn}},\
  }\href {https://doi.org/10.1038/35051009} {\bibfield  {journal} {\bibinfo
  {journal} {Nature}\ }\textbf {\bibinfo {volume} {409}},\ \bibinfo {pages}
  {46} (\bibinfo {year} {2001})}\BibitemShut {NoStop}%
\bibitem [{\citenamefont {Kok}\ \emph {et~al.}(2007)\citenamefont {Kok},
  \citenamefont {Munro}, \citenamefont {Nemoto}, \citenamefont {Ralph},
  \citenamefont {Dowling},\ and\ \citenamefont {Milburn}}]{Kok2007}%
  \BibitemOpen
  \bibfield  {author} {\bibinfo {author} {\bibfnamefont {P.}~\bibnamefont
  {Kok}}, \bibinfo {author} {\bibfnamefont {W.~J.}\ \bibnamefont {Munro}},
  \bibinfo {author} {\bibfnamefont {K.}~\bibnamefont {Nemoto}}, \bibinfo
  {author} {\bibfnamefont {T.~C.}\ \bibnamefont {Ralph}}, \bibinfo {author}
  {\bibfnamefont {J.~P.}\ \bibnamefont {Dowling}}, \ and\ \bibinfo {author}
  {\bibfnamefont {G.~J.}\ \bibnamefont {Milburn}},\ }\href
  {https://doi.org/10.1103/RevModPhys.79.135} {\bibfield  {journal} {\bibinfo
  {journal} {Rev. Mod. Phys.}\ }\textbf {\bibinfo {volume} {79}},\ \bibinfo
  {pages} {135} (\bibinfo {year} {2007})}\BibitemShut {NoStop}%
\bibitem [{\citenamefont {Aaronson}\ and\ \citenamefont
  {Arkhipov}(2011)}]{Aaronson2011}%
  \BibitemOpen
  \bibfield  {author} {\bibinfo {author} {\bibfnamefont {S.}~\bibnamefont
  {Aaronson}}\ and\ \bibinfo {author} {\bibfnamefont {A.}~\bibnamefont
  {Arkhipov}},\ }in\ \href@noop {} {\emph {\bibinfo {booktitle} {Proceedings of
  the forty-third annual ACM symposium on Theory of computing}}}\ (\bibinfo
  {organization} {ACM},\ \bibinfo {year} {2011})\ pp.\ \bibinfo {pages}
  {333--342}\BibitemShut {NoStop}%
\bibitem [{\citenamefont {Milburn}\ \emph {et~al.}(1997)\citenamefont
  {Milburn}, \citenamefont {Corney}, \citenamefont {Wright},\ and\
  \citenamefont {Walls}}]{Milburn1997}%
  \BibitemOpen
  \bibfield  {author} {\bibinfo {author} {\bibfnamefont {G.}~\bibnamefont
  {Milburn}}, \bibinfo {author} {\bibfnamefont {J.}~\bibnamefont {Corney}},
  \bibinfo {author} {\bibfnamefont {E.~M.}\ \bibnamefont {Wright}}, \ and\
  \bibinfo {author} {\bibfnamefont {D.}~\bibnamefont {Walls}},\ }\href
  {https://doi.org/10.1103/PhysRevA.55.4318} {\bibfield  {journal} {\bibinfo
  {journal} {Phys. Rev. A}\ }\textbf {\bibinfo {volume} {55}},\ \bibinfo
  {pages} {4318} (\bibinfo {year} {1997})}\BibitemShut {NoStop}%
\bibitem [{\citenamefont {Albiez}\ \emph {et~al.}(2005)\citenamefont {Albiez},
  \citenamefont {Gati}, \citenamefont {F{\"o}lling}, \citenamefont {Hunsmann},
  \citenamefont {Cristiani},\ and\ \citenamefont {Oberthaler}}]{Albiez2005}%
  \BibitemOpen
  \bibfield  {author} {\bibinfo {author} {\bibfnamefont {M.}~\bibnamefont
  {Albiez}}, \bibinfo {author} {\bibfnamefont {R.}~\bibnamefont {Gati}},
  \bibinfo {author} {\bibfnamefont {J.}~\bibnamefont {F{\"o}lling}}, \bibinfo
  {author} {\bibfnamefont {S.}~\bibnamefont {Hunsmann}}, \bibinfo {author}
  {\bibfnamefont {M.}~\bibnamefont {Cristiani}}, \ and\ \bibinfo {author}
  {\bibfnamefont {M.~K.}\ \bibnamefont {Oberthaler}},\ }\href
  {https://doi.org/10.1103/PhysRevLett.95.010402} {\bibfield  {journal}
  {\bibinfo  {journal} {Phys. Rev. Lett.}\ }\textbf {\bibinfo {volume} {95}},\
  \bibinfo {pages} {010402} (\bibinfo {year} {2005})}\BibitemShut {NoStop}%
\bibitem [{\citenamefont {Terhal}\ and\ \citenamefont
  {DiVincenzo}(2002)}]{Terhal2002}%
  \BibitemOpen
  \bibfield  {author} {\bibinfo {author} {\bibfnamefont {B.~M.}\ \bibnamefont
  {Terhal}}\ and\ \bibinfo {author} {\bibfnamefont {D.~P.}\ \bibnamefont
  {DiVincenzo}},\ }\href {https://doi.org/10.1103/PhysRevA.65.032325}
  {\bibfield  {journal} {\bibinfo  {journal} {Phys. Rev. A}\ }\textbf {\bibinfo
  {volume} {65}},\ \bibinfo {pages} {032325} (\bibinfo {year}
  {2002})}\BibitemShut {NoStop}%
\bibitem [{\citenamefont {DiVincenzo}\ and\ \citenamefont
  {Terhal}(2005)}]{DiVincenzo2005}%
  \BibitemOpen
  \bibfield  {author} {\bibinfo {author} {\bibfnamefont {D.~P.}\ \bibnamefont
  {DiVincenzo}}\ and\ \bibinfo {author} {\bibfnamefont {B.~M.}\ \bibnamefont
  {Terhal}},\ }\href {http://dx.doi.org/10.1007/s10701-005-8657-0} {\bibfield
  {journal} {\bibinfo  {journal} {Found. Phys.}\ }\textbf {\bibinfo {volume}
  {35}},\ \bibinfo {pages} {1967} (\bibinfo {year} {2005})}\BibitemShut
  {NoStop}%
\bibitem [{\citenamefont {Bocquillon}\ \emph {et~al.}(2014)\citenamefont
  {Bocquillon}, \citenamefont {Freulon}, \citenamefont {Parmentier},
  \citenamefont {Berroir}, \citenamefont {Pla{\c{c}}ais}, \citenamefont {Wahl},
  \citenamefont {Rech}, \citenamefont {Jonckheere}, \citenamefont {Martin},
  \citenamefont {Grenier} \emph {et~al.}}]{Bocquillon2014}%
  \BibitemOpen
  \bibfield  {author} {\bibinfo {author} {\bibfnamefont {E.}~\bibnamefont
  {Bocquillon}}, \bibinfo {author} {\bibfnamefont {V.}~\bibnamefont {Freulon}},
  \bibinfo {author} {\bibfnamefont {F.~D.}\ \bibnamefont {Parmentier}},
  \bibinfo {author} {\bibfnamefont {J.-M.}\ \bibnamefont {Berroir}}, \bibinfo
  {author} {\bibfnamefont {B.}~\bibnamefont {Pla{\c{c}}ais}}, \bibinfo {author}
  {\bibfnamefont {C.}~\bibnamefont {Wahl}}, \bibinfo {author} {\bibfnamefont
  {J.}~\bibnamefont {Rech}}, \bibinfo {author} {\bibfnamefont {T.}~\bibnamefont
  {Jonckheere}}, \bibinfo {author} {\bibfnamefont {T.}~\bibnamefont {Martin}},
  \bibinfo {author} {\bibfnamefont {C.}~\bibnamefont {Grenier}},  \emph
  {et~al.},\ }\href@noop {} {\bibfield  {journal} {\bibinfo  {journal} {Annalen
  der Physik}\ }\textbf {\bibinfo {volume} {526}},\ \bibinfo {pages} {1}
  (\bibinfo {year} {2014})}\BibitemShut {NoStop}%
\bibitem [{\citenamefont {Tichy}(2014)}]{Tichy2014}%
  \BibitemOpen
  \bibfield  {author} {\bibinfo {author} {\bibfnamefont {M.~C.}\ \bibnamefont
  {Tichy}},\ }\href@noop {} {\bibfield  {journal} {\bibinfo  {journal}
  {J.~Phys.~B}\ }\textbf {\bibinfo {volume} {47}},\ \bibinfo {pages} {103001}
  (\bibinfo {year} {2014})}\BibitemShut {NoStop}%
\bibitem [{\citenamefont {Dasenbrook}\ \emph {et~al.}(2016)\citenamefont
  {Dasenbrook}, \citenamefont {Bowles}, \citenamefont {Brask}, \citenamefont
  {Hofer}, \citenamefont {Flindt},\ and\ \citenamefont
  {Brunner}}]{Dasenbrook2016}%
  \BibitemOpen
  \bibfield  {author} {\bibinfo {author} {\bibfnamefont {D.}~\bibnamefont
  {Dasenbrook}}, \bibinfo {author} {\bibfnamefont {J.}~\bibnamefont {Bowles}},
  \bibinfo {author} {\bibfnamefont {J.~B.}\ \bibnamefont {Brask}}, \bibinfo
  {author} {\bibfnamefont {P.~P.}\ \bibnamefont {Hofer}}, \bibinfo {author}
  {\bibfnamefont {C.}~\bibnamefont {Flindt}}, \ and\ \bibinfo {author}
  {\bibfnamefont {N.}~\bibnamefont {Brunner}},\ }\href@noop {} {\bibfield
  {journal} {\bibinfo  {journal} {New Journal of Physics}\ }\textbf {\bibinfo
  {volume} {18}},\ \bibinfo {pages} {043036} (\bibinfo {year}
  {2016})}\BibitemShut {NoStop}%
\bibitem [{\citenamefont {Bravyi}\ and\ \citenamefont
  {Koenig}(2011)}]{Bravyi2011}%
  \BibitemOpen
  \bibfield  {author} {\bibinfo {author} {\bibfnamefont {S.}~\bibnamefont
  {Bravyi}}\ and\ \bibinfo {author} {\bibfnamefont {R.}~\bibnamefont
  {Koenig}},\ }\href@noop {} {\bibfield  {journal} {\bibinfo  {journal}
  {preprint arXiv:1112.2184}\ } (\bibinfo {year} {2011})}\BibitemShut {NoStop}%
\bibitem [{\citenamefont {de~Melo}\ \emph {et~al.}(2013)\citenamefont
  {de~Melo}, \citenamefont {{\'C}wikli{\'n}ski},\ and\ \citenamefont
  {Terhal}}]{Melo2013}%
  \BibitemOpen
  \bibfield  {author} {\bibinfo {author} {\bibfnamefont {F.}~\bibnamefont
  {de~Melo}}, \bibinfo {author} {\bibfnamefont {P.}~\bibnamefont
  {{\'C}wikli{\'n}ski}}, \ and\ \bibinfo {author} {\bibfnamefont {B.~M.}\
  \bibnamefont {Terhal}},\ }\href@noop {} {\bibfield  {journal} {\bibinfo
  {journal} {New Journal of Physics}\ }\textbf {\bibinfo {volume} {15}},\
  \bibinfo {pages} {013015} (\bibinfo {year} {2013})}\BibitemShut {NoStop}%
\bibitem [{\citenamefont {Oszmaniec}\ \emph {et~al.}(2014)\citenamefont
  {Oszmaniec}, \citenamefont {Gutt},\ and\ \citenamefont
  {Ku{\'s}}}]{Oszmaniec2014}%
  \BibitemOpen
  \bibfield  {author} {\bibinfo {author} {\bibfnamefont {M.}~\bibnamefont
  {Oszmaniec}}, \bibinfo {author} {\bibfnamefont {J.}~\bibnamefont {Gutt}}, \
  and\ \bibinfo {author} {\bibfnamefont {M.}~\bibnamefont {Ku{\'s}}},\
  }\href@noop {} {\bibfield  {journal} {\bibinfo  {journal} {Phys. Rev. A}\
  }\textbf {\bibinfo {volume} {90}},\ \bibinfo {pages} {020302} (\bibinfo
  {year} {2014})}\BibitemShut {NoStop}%
\bibitem [{\citenamefont {Valiant}(2002)}]{Valiant2002}%
  \BibitemOpen
  \bibfield  {author} {\bibinfo {author} {\bibfnamefont {L.~G.}\ \bibnamefont
  {Valiant}},\ }\href@noop {} {\bibfield  {journal} {\bibinfo  {journal} {SIAM
  Journal on Computing}\ }\textbf {\bibinfo {volume} {31}},\ \bibinfo {pages}
  {1229} (\bibinfo {year} {2002})}\BibitemShut {NoStop}%
\bibitem [{\citenamefont {Knill}(2001)}]{Knill2001a}%
  \BibitemOpen
  \bibfield  {author} {\bibinfo {author} {\bibfnamefont {E.}~\bibnamefont
  {Knill}},\ }\href@noop {} {\bibfield  {journal} {\bibinfo  {journal} {arXiv
  preprint quant-ph/0108033}\ } (\bibinfo {year} {2001})}\BibitemShut {NoStop}%
\bibitem [{\citenamefont {Jozsa}\ and\ \citenamefont
  {Miyake}(2008)}]{Jozsa2008}%
  \BibitemOpen
  \bibfield  {author} {\bibinfo {author} {\bibfnamefont {R.}~\bibnamefont
  {Jozsa}}\ and\ \bibinfo {author} {\bibfnamefont {A.}~\bibnamefont {Miyake}},\
  }in\ \href@noop {} {\emph {\bibinfo {booktitle} {Proceedings of the Royal
  Society of London A: Mathematical, Physical and Engineering Sciences}}},\
  Vol.\ \bibinfo {volume} {464}\ (\bibinfo {organization} {The Royal Society},\
  \bibinfo {year} {2008})\ pp.\ \bibinfo {pages} {3089--3106}\BibitemShut
  {NoStop}%
\bibitem [{\citenamefont {Brod}(2016)}]{Brod2016}%
  \BibitemOpen
  \bibfield  {author} {\bibinfo {author} {\bibfnamefont {D.~J.}\ \bibnamefont
  {Brod}},\ }\href {\doibase 10.1103/PhysRevA.93.062332} {\bibfield  {journal}
  {\bibinfo  {journal} {Phys. Rev. A}\ }\textbf {\bibinfo {volume} {93}},\
  \bibinfo {pages} {062332} (\bibinfo {year} {2016})}\BibitemShut {NoStop}%
\bibitem [{\citenamefont {Bouland}\ and\ \citenamefont
  {Aaronson}(2014)}]{Bouland2014}%
  \BibitemOpen
  \bibfield  {author} {\bibinfo {author} {\bibfnamefont {A.}~\bibnamefont
  {Bouland}}\ and\ \bibinfo {author} {\bibfnamefont {S.}~\bibnamefont
  {Aaronson}},\ }\href@noop {} {\bibfield  {journal} {\bibinfo  {journal}
  {Physical Review A}\ }\textbf {\bibinfo {volume} {89}},\ \bibinfo {pages}
  {062316} (\bibinfo {year} {2014})}\BibitemShut {NoStop}%
\bibitem [{\citenamefont {Aaronson}\ \emph {et~al.}(2015)\citenamefont
  {Aaronson}, \citenamefont {Grier},\ and\ \citenamefont
  {Schaeffer}}]{Aaronson2015}%
  \BibitemOpen
  \bibfield  {author} {\bibinfo {author} {\bibfnamefont {S.}~\bibnamefont
  {Aaronson}}, \bibinfo {author} {\bibfnamefont {D.}~\bibnamefont {Grier}}, \
  and\ \bibinfo {author} {\bibfnamefont {L.}~\bibnamefont {Schaeffer}},\
  }\href@noop {} {\bibfield  {journal} {\bibinfo  {journal} {arXiv preprint
  arXiv:1504.05155}\ } (\bibinfo {year} {2015})}\BibitemShut {NoStop}%
\bibitem [{\citenamefont {Altafini}(2002)}]{Altafini2002}%
  \BibitemOpen
  \bibfield  {author} {\bibinfo {author} {\bibfnamefont {C.}~\bibnamefont
  {Altafini}},\ }\href {http://doi.org/10.1063/1.1467611} {\bibfield  {journal}
  {\bibinfo  {journal} {J. Math. Phys.}\ }\textbf {\bibinfo {volume} {43}},\
  \bibinfo {pages} {2051} (\bibinfo {year} {2002})}\BibitemShut {NoStop}%
\bibitem [{\citenamefont {Zeier}\ and\ \citenamefont
  {Schulte-Herbr{\"u}ggen}(2011)}]{Zeier2011}%
  \BibitemOpen
  \bibfield  {author} {\bibinfo {author} {\bibfnamefont {R.}~\bibnamefont
  {Zeier}}\ and\ \bibinfo {author} {\bibfnamefont {T.}~\bibnamefont
  {Schulte-Herbr{\"u}ggen}},\ }\href {http://doi.org/10.1063/1.3657939}
  {\bibfield  {journal} {\bibinfo  {journal} {J. Math. Phys.}\ }\textbf
  {\bibinfo {volume} {52}},\ \bibinfo {pages} {113510} (\bibinfo {year}
  {2011})}\BibitemShut {NoStop}%
\bibitem [{\citenamefont {Zimbor{\'a}s}\ \emph {et~al.}(2014)\citenamefont
  {Zimbor{\'a}s}, \citenamefont {Zeier}, \citenamefont {Keyl},\ and\
  \citenamefont {Schulte-Herbr{\"u}ggen}}]{ZZKS2014}%
  \BibitemOpen
  \bibfield  {author} {\bibinfo {author} {\bibfnamefont {Z.}~\bibnamefont
  {Zimbor{\'a}s}}, \bibinfo {author} {\bibfnamefont {R.}~\bibnamefont {Zeier}},
  \bibinfo {author} {\bibfnamefont {M.}~\bibnamefont {Keyl}}, \ and\ \bibinfo
  {author} {\bibfnamefont {T.}~\bibnamefont {Schulte-Herbr{\"u}ggen}},\ }\href
  {http://doi.org/10.1063/1.3657939} {\bibfield  {journal} {\bibinfo  {journal}
  {EPJ Quantum Technology}\ }\textbf {\bibinfo {volume} {1}},\ \bibinfo {pages}
  {11} (\bibinfo {year} {2014})}\BibitemShut {NoStop}%
\bibitem [{\citenamefont {Zimbor{\'a}s}\ \emph {et~al.}(2015)\citenamefont
  {Zimbor{\'a}s}, \citenamefont {Zeier}, \citenamefont
  {Schulte-Herbr{\"u}ggen},\ and\ \citenamefont {Burgarth}}]{ZZSB2015}%
  \BibitemOpen
  \bibfield  {author} {\bibinfo {author} {\bibfnamefont {Z.}~\bibnamefont
  {Zimbor{\'a}s}}, \bibinfo {author} {\bibfnamefont {R.}~\bibnamefont {Zeier}},
  \bibinfo {author} {\bibfnamefont {T.}~\bibnamefont {Schulte-Herbr{\"u}ggen}},
  \ and\ \bibinfo {author} {\bibfnamefont {D.}~\bibnamefont {Burgarth}},\
  }\href {https://doi.org/10.1103/PhysRevA.92.042309} {\bibfield  {journal}
  {\bibinfo  {journal} {Phys. Rev. A}\ }\textbf {\bibinfo {volume} {92}},\
  \bibinfo {pages} {042309} (\bibinfo {year} {2015})}\BibitemShut {NoStop}%
\bibitem [{\citenamefont {Sawicki}(2016)}]{Sawicki2016}%
  \BibitemOpen
  \bibfield  {author} {\bibinfo {author} {\bibfnamefont {A.}~\bibnamefont
  {Sawicki}},\ }\href
  {http://www.rintonpress.com/xxqic16/qic-16-34/0291-0312.pdf} {\bibfield
  {journal} {\bibinfo  {journal} {Quantum Inf. Comput.}\ }\textbf {\bibinfo
  {volume} {16}},\ \bibinfo {pages} {0291} (\bibinfo {year}
  {2016})}\BibitemShut {NoStop}%
\bibitem [{\citenamefont {Sawicki}\ and\ \citenamefont
  {Karnas}(2016)}]{GatesSawickiLong2016}%
  \BibitemOpen
  \bibfield  {author} {\bibinfo {author} {\bibfnamefont {A.}~\bibnamefont
  {Sawicki}}\ and\ \bibinfo {author} {\bibfnamefont {K.}~\bibnamefont
  {Karnas}},\ }\href@noop {} {\bibfield  {journal} {\bibinfo  {journal} {arXiv
  preprint quant-ph/1609.05780}\ } (\bibinfo {year} {2016})}\BibitemShut
  {NoStop}%
\bibitem [{\citenamefont {Brod}\ and\ \citenamefont
  {Combes}(2016)}]{Brod2016a}%
  \BibitemOpen
  \bibfield  {author} {\bibinfo {author} {\bibfnamefont {D.~J.}\ \bibnamefont
  {Brod}}\ and\ \bibinfo {author} {\bibfnamefont {J.}~\bibnamefont {Combes}},\
  }\href {\doibase 10.1103/PhysRevLett.117.080502} {\bibfield  {journal}
  {\bibinfo  {journal} {Phys. Rev. Lett.}\ }\textbf {\bibinfo {volume} {117}},\
  \bibinfo {pages} {080502} (\bibinfo {year} {2016})}\BibitemShut {NoStop}%
\bibitem [{\citenamefont {Oszmaniec}\ \emph {et~al.}(2016)\citenamefont
  {Oszmaniec}, \citenamefont {Augusiak}, \citenamefont {Gogolin}, \citenamefont
  {Ko\l{}ody\ifmmode~\acute{n}\else \'{n}\fi{}ski}, \citenamefont {Ac\'{\i}n},\
  and\ \citenamefont {Lewenstein}}]{Oszmaniec2016}%
  \BibitemOpen
  \bibfield  {author} {\bibinfo {author} {\bibfnamefont {M.}~\bibnamefont
  {Oszmaniec}}, \bibinfo {author} {\bibfnamefont {R.}~\bibnamefont {Augusiak}},
  \bibinfo {author} {\bibfnamefont {C.}~\bibnamefont {Gogolin}}, \bibinfo
  {author} {\bibfnamefont {J.}~\bibnamefont {Ko\l{}ody\ifmmode~\acute{n}\else
  \'{n}\fi{}ski}}, \bibinfo {author} {\bibfnamefont {A.}~\bibnamefont
  {Ac\'{\i}n}}, \ and\ \bibinfo {author} {\bibfnamefont {M.}~\bibnamefont
  {Lewenstein}},\ }\href {\doibase 10.1103/PhysRevX.6.041044} {\bibfield
  {journal} {\bibinfo  {journal} {Phys. Rev. X}\ }\textbf {\bibinfo {volume}
  {6}},\ \bibinfo {pages} {041044} (\bibinfo {year} {2016})}\BibitemShut
  {NoStop}%
\bibitem [{\citenamefont {Bravyi}(2004)}]{Bravyi2004}%
  \BibitemOpen
  \bibfield  {author} {\bibinfo {author} {\bibfnamefont {S.}~\bibnamefont
  {Bravyi}},\ }\href@noop {} {\bibfield  {journal} {\bibinfo  {journal} {arXiv
  preprint quant-ph/0404180}\ } (\bibinfo {year} {2004})}\BibitemShut {NoStop}%
\bibitem [{\citenamefont {Hall}(2000)}]{HallGroups}%
  \BibitemOpen
  \bibfield  {author} {\bibinfo {author} {\bibfnamefont {B.~C.}\ \bibnamefont
  {Hall}},\ }\href@noop {} {\emph {\bibinfo {title} {{L}ie Groups, Lie
  Algebras, and Representations: An Elementary Introduction}}}\ (\bibinfo
  {publisher} {Springer},\ \bibinfo {year} {2000})\BibitemShut {NoStop}%
\bibitem [{\citenamefont {Dynkin}\ \emph {et~al.}(2000)\citenamefont {Dynkin},
  \citenamefont {Seitz},\ and\ \citenamefont {Onishchik}}]{Dynkin2000}%
  \BibitemOpen
  \bibfield  {author} {\bibinfo {author} {\bibfnamefont {E.~B.}\ \bibnamefont
  {Dynkin}}, \bibinfo {author} {\bibfnamefont {G.~M.}\ \bibnamefont {Seitz}}, \
  and\ \bibinfo {author} {\bibfnamefont {A.~L.}\ \bibnamefont {Onishchik}},\
  }\href@noop {} {\emph {\bibinfo {title} {Selected papers of E. B. Dynkin with
  commentary}}}\ (\bibinfo  {publisher} {Americal Mathematical Society},\
  \bibinfo {year} {2000})\ \bibinfo {note} {in Chapter "Maximall subgroups of
  classical groups"}\BibitemShut {NoStop}%
\bibitem [{\citenamefont {Brand\~ao}\ \emph {et~al.}(2016)\citenamefont
  {Brand\~ao}, \citenamefont {Harrow},\ and\ \citenamefont
  {Horodecki}}]{Horo2016}%
  \BibitemOpen
  \bibfield  {author} {\bibinfo {author} {\bibfnamefont {F.~G. S.~L.}\
  \bibnamefont {Brand\~ao}}, \bibinfo {author} {\bibfnamefont {A.~W.}\
  \bibnamefont {Harrow}}, \ and\ \bibinfo {author} {\bibfnamefont
  {M.}~\bibnamefont {Horodecki}},\ }\href {\doibase
  10.1103/PhysRevLett.116.170502} {\bibfield  {journal} {\bibinfo  {journal}
  {Phys. Rev. Lett.}\ }\textbf {\bibinfo {volume} {116}},\ \bibinfo {pages}
  {170502} (\bibinfo {year} {2016})}\BibitemShut {NoStop}%
\bibitem [{\citenamefont {Lloyd}\ and\ \citenamefont
  {Braunstein}(1999)}]{lloydcont1999}%
  \BibitemOpen
  \bibfield  {author} {\bibinfo {author} {\bibfnamefont {S.}~\bibnamefont
  {Lloyd}}\ and\ \bibinfo {author} {\bibfnamefont {S.~L.}\ \bibnamefont
  {Braunstein}},\ }\href {\doibase 10.1103/PhysRevLett.82.1784} {\bibfield
  {journal} {\bibinfo  {journal} {Phys. Rev. Lett.}\ }\textbf {\bibinfo
  {volume} {82}},\ \bibinfo {pages} {1784} (\bibinfo {year}
  {1999})}\BibitemShut {NoStop}%
\bibitem [{\citenamefont {Schirmer}\ \emph {et~al.}(2002)\citenamefont
  {Schirmer}, \citenamefont {Solomon},\ and\ \citenamefont
  {Leahy}}]{schirmer2002}%
  \BibitemOpen
  \bibfield  {author} {\bibinfo {author} {\bibfnamefont {S.~G.}\ \bibnamefont
  {Schirmer}}, \bibinfo {author} {\bibfnamefont {A.~I.}\ \bibnamefont
  {Solomon}}, \ and\ \bibinfo {author} {\bibfnamefont {J.~V.}\ \bibnamefont
  {Leahy}},\ }\href@noop {} {\bibfield  {journal} {\bibinfo  {journal} {J.
  Phys. A}\ }\textbf {\bibinfo {volume} {35}},\ \bibinfo {pages} {4125}
  (\bibinfo {year} {2002})}\BibitemShut {NoStop}%
\bibitem [{\citenamefont {Albertini}\ and\ \citenamefont
  {D'Alessandro}(2003)}]{albertini2003}%
  \BibitemOpen
  \bibfield  {author} {\bibinfo {author} {\bibfnamefont {F.}~\bibnamefont
  {Albertini}}\ and\ \bibinfo {author} {\bibfnamefont {D.}~\bibnamefont
  {D'Alessandro}},\ }\href@noop {} {\bibfield  {journal} {\bibinfo  {journal}
  {IEEE Trans. Aut. Cont.}\ }\textbf {\bibinfo {volume} {48}},\ \bibinfo
  {pages} {1399} (\bibinfo {year} {2003})}\BibitemShut {NoStop}%
\bibitem [{\citenamefont {Greiter}\ \emph {et~al.}(1991)\citenamefont
  {Greiter}, \citenamefont {Wen},\ and\ \citenamefont {Wilczek}}]{Greiter1991}%
  \BibitemOpen
  \bibfield  {author} {\bibinfo {author} {\bibfnamefont {M.}~\bibnamefont
  {Greiter}}, \bibinfo {author} {\bibfnamefont {X.-G.}\ \bibnamefont {Wen}}, \
  and\ \bibinfo {author} {\bibfnamefont {F.}~\bibnamefont {Wilczek}},\ }\href
  {\doibase 10.1103/PhysRevLett.66.3205} {\bibfield  {journal} {\bibinfo
  {journal} {Phys. Rev. Lett.}\ }\textbf {\bibinfo {volume} {66}},\ \bibinfo
  {pages} {3205} (\bibinfo {year} {1991})}\BibitemShut {NoStop}%
\bibitem [{\citenamefont {Dolcini}\ and\ \citenamefont
  {Montorsi}(2013)}]{DM2013}%
  \BibitemOpen
  \bibfield  {author} {\bibinfo {author} {\bibfnamefont {F.}~\bibnamefont
  {Dolcini}}\ and\ \bibinfo {author} {\bibfnamefont {A.}~\bibnamefont
  {Montorsi}},\ }\href@noop {} {\bibfield  {journal} {\bibinfo  {journal}
  {Phys. Rev. B}\ }\textbf {\bibinfo {volume} {88}},\ \bibinfo {pages} {115115}
  (\bibinfo {year} {2013})}\BibitemShut {NoStop}%
\bibitem [{\citenamefont {Brod}\ and\ \citenamefont {Childs}(2014)}]{Brod2014}%
  \BibitemOpen
  \bibfield  {author} {\bibinfo {author} {\bibfnamefont {D.~J.}\ \bibnamefont
  {Brod}}\ and\ \bibinfo {author} {\bibfnamefont {A.~M.}\ \bibnamefont
  {Childs}},\ }\href {https://arxiv.org/pdf/1308.1463.pdf} {\bibfield
  {journal} {\bibinfo  {journal} {Quantum Inf. Comput.}\ }\textbf {\bibinfo
  {volume} {14}},\ \bibinfo {pages} {901} (\bibinfo {year} {2014})}\BibitemShut
  {NoStop}%
\bibitem [{\citenamefont {Bravyi}\ and\ \citenamefont
  {Gosset}(2016)}]{Bravyi2016}%
  \BibitemOpen
  \bibfield  {author} {\bibinfo {author} {\bibfnamefont {S.}~\bibnamefont
  {Bravyi}}\ and\ \bibinfo {author} {\bibfnamefont {D.}~\bibnamefont
  {Gosset}},\ }\href {https://arxiv.org/pdf/1609.00735.pdf} {\bibfield
  {journal} {\bibinfo  {journal} {arXiv preprint arXiv:1609.00735}\ } (\bibinfo
  {year} {2016})}\BibitemShut {NoStop}%
\bibitem [{\citenamefont {Nielsen}\ \emph {et~al.}(2006)\citenamefont
  {Nielsen}, \citenamefont {Dowling}, \citenamefont {Gu},\ and\ \citenamefont
  {Doherty}}]{Nielsen2006}%
  \BibitemOpen
  \bibfield  {author} {\bibinfo {author} {\bibfnamefont {M.~A.}\ \bibnamefont
  {Nielsen}}, \bibinfo {author} {\bibfnamefont {M.~R.}\ \bibnamefont
  {Dowling}}, \bibinfo {author} {\bibfnamefont {M.}~\bibnamefont {Gu}}, \ and\
  \bibinfo {author} {\bibfnamefont {A.~C.}\ \bibnamefont {Doherty}},\ }\href
  {\doibase 10.1126/science.1121541} {\bibfield  {journal} {\bibinfo  {journal}
  {Science}\ }\textbf {\bibinfo {volume} {311}},\ \bibinfo {pages} {1133}
  (\bibinfo {year} {2006})}\BibitemShut {NoStop}%
\bibitem [{\citenamefont {Harrow}\ \emph {et~al.}(2002)\citenamefont {Harrow},
  \citenamefont {Recht},\ and\ \citenamefont {Chuang}}]{Harrow2002}%
  \BibitemOpen
  \bibfield  {author} {\bibinfo {author} {\bibfnamefont {A.~W.}\ \bibnamefont
  {Harrow}}, \bibinfo {author} {\bibfnamefont {B.}~\bibnamefont {Recht}}, \
  and\ \bibinfo {author} {\bibfnamefont {I.~L.}\ \bibnamefont {Chuang}},\
  }\href {\doibase 10.1063/1.1495899} {\bibfield  {journal} {\bibinfo
  {journal} {J.~Math.~Phys.}\ }\textbf {\bibinfo {volume} {43}},\ \bibinfo
  {pages} {4445} (\bibinfo {year} {2002})},\ \Eprint
  {http://arxiv.org/abs/http://dx.doi.org/10.1063/1.1495899}
  {http://dx.doi.org/10.1063/1.1495899} \BibitemShut {NoStop}%
\bibitem [{\citenamefont {Fulton}\ and\ \citenamefont
  {Harris}(1991)}]{FultonHarris}%
  \BibitemOpen
  \bibfield  {author} {\bibinfo {author} {\bibfnamefont {W.}~\bibnamefont
  {Fulton}}\ and\ \bibinfo {author} {\bibfnamefont {J.}~\bibnamefont
  {Harris}},\ }\href@noop {} {\emph {\bibinfo {title} {Representation {T}heory:
  {A} {F}irst {C}ourse}}}\ (\bibinfo  {publisher} {Springer},\ \bibinfo {year}
  {1991})\BibitemShut {NoStop}%
\bibitem [{\citenamefont {Humphreys}(1972)}]{Humphreys1972}%
  \BibitemOpen
  \bibfield  {author} {\bibinfo {author} {\bibfnamefont {J.}~\bibnamefont
  {Humphreys}},\ }\href {https://books.google.hu/books?id=TiUlAQAAIAAJ} {\emph
  {\bibinfo {title} {Introduction to Lie Algebras and Representation
  Theory}}},\ Graduate Texts in Mathematics Series\ (\bibinfo  {publisher}
  {Springer-Verlag GmbH},\ \bibinfo {year} {1972})\BibitemShut {NoStop}%
\bibitem [{\citenamefont {Oszmaniec}(2014)}]{Oszmaniec2014c}%
  \BibitemOpen
  \bibfield  {author} {\bibinfo {author} {\bibfnamefont {M.}~\bibnamefont
  {Oszmaniec}},\ }\href@noop {} {\bibfield  {journal} {\bibinfo  {journal}
  {arXiv preprint arXiv:1412.4657}\ } (\bibinfo {year} {2014})}\BibitemShut
  {NoStop}%
\bibitem [{\citenamefont {Sakurai}\ and\ \citenamefont
  {Napolitano}(2011)}]{Sakurai2011}%
  \BibitemOpen
  \bibfield  {author} {\bibinfo {author} {\bibfnamefont {J.~J.}\ \bibnamefont
  {Sakurai}}\ and\ \bibinfo {author} {\bibfnamefont {J.}~\bibnamefont
  {Napolitano}},\ }\href@noop {} {\emph {\bibinfo {title} {Modern quantum
  mechanics}}}\ (\bibinfo  {publisher} {Addison-Wesley},\ \bibinfo {year}
  {2011})\BibitemShut {NoStop}%
\bibitem [{\citenamefont {Macdonald}(1998)}]{symmetricFunctions1998}%
  \BibitemOpen
  \bibfield  {author} {\bibinfo {author} {\bibfnamefont {I.~G.}\ \bibnamefont
  {Macdonald}},\ }\href@noop {} {\emph {\bibinfo {title} {Symmetric Functions
  and Orthogonal Polynomials}}}\ (\bibinfo  {publisher} {Americal Mathematical
  Society},\ \bibinfo {year} {1998})\BibitemShut {NoStop}%
\bibitem [{\citenamefont {Zeier}\ and\ \citenamefont
  {Zimborás}(2015)}]{ZZ2015}%
  \BibitemOpen
  \bibfield  {author} {\bibinfo {author} {\bibfnamefont {R.}~\bibnamefont
  {Zeier}}\ and\ \bibinfo {author} {\bibfnamefont {Z.}~\bibnamefont
  {Zimborás}},\ }\href {http://dx.doi.org/10.1063/1.4928410} {\bibfield
  {journal} {\bibinfo  {journal} {Journal of Mathematical Physics}\ }\textbf
  {\bibinfo {volume} {56}},\ \bibinfo {pages} {081702} (\bibinfo {year}
  {2015})}\BibitemShut {NoStop}%
\bibitem [{\citenamefont {Hofmann}\ and\ \citenamefont
  {Morris}(2013)}]{hofmann2013}%
  \BibitemOpen
  \bibfield  {author} {\bibinfo {author} {\bibfnamefont {K.~H.}\ \bibnamefont
  {Hofmann}}\ and\ \bibinfo {author} {\bibfnamefont {S.~A.}\ \bibnamefont
  {Morris}},\ }\href@noop {} {\emph {\bibinfo {title} {The structure of compact
  groups: a primer for the student-a handbook for the expert}}},\ Vol.~\bibinfo
  {volume} {25}\ (\bibinfo  {publisher} {Walter de Gruyter},\ \bibinfo {year}
  {2013})\BibitemShut {NoStop}%
\bibitem [{Note1()}]{Note1}%
  \BibitemOpen
  \bibinfo {note} {The exception is the case $\protect \mathfrak {k}=so(4)$
  appearing for $\protect \mathrm {FLO}$ for $d=2$ modes. For this case we have
  $\pi ^+_\protect \mathrm {FLO}(\protect \mathfrak {so}(4))\approx \protect
  \mathfrak {su}(\protect \mathcal {H}^{+}_{\protect \mathrm
  {Fock}})$.}\BibitemShut {Stop}%
\bibitem [{Note2()}]{Note2}%
  \BibitemOpen
  \bibinfo {note} {In order to prove this assertion we notice that from $\tau
  _{V'} = \protect \mathrm {Ad}_{U'} \alpha $ and $\protect \mathrm {Ad}_U
  \alpha $ it follows that $ \protect \mathrm {Ad}_{U'^\dagger } \tau _{V'}=
  \protect \mathrm {Ad}_{U^\dagger } \tau _{V}$. We can now use the reasoning
  analogous to the one given in point (I) to prove that $V' = \protect \mathrm
  {exp}(\protect \mathrm {i}\theta ) \Pi _f(U U'^\dagger ) V$.}\BibitemShut
  {Stop}%
\bibitem [{\citenamefont {Oszmaniec}\ and\ \citenamefont
  {Kuś}(2012)}]{Oszmaniec2012}%
  \BibitemOpen
  \bibfield  {author} {\bibinfo {author} {\bibfnamefont {M.}~\bibnamefont
  {Oszmaniec}}\ and\ \bibinfo {author} {\bibfnamefont {M.}~\bibnamefont
  {Kuś}},\ }\href {http://stacks.iop.org/1751-8121/45/i=24/a=244034}
  {\bibfield  {journal} {\bibinfo  {journal} {Journal of Physics A:
  Mathematical and Theoretical}\ }\textbf {\bibinfo {volume} {45}},\ \bibinfo
  {pages} {244034} (\bibinfo {year} {2012})}\BibitemShut {NoStop}%
\end{thebibliography}%

\appendix

\onecolumngrid
\newpage
\appendix

\part*{Appendices}

In the appendices we provide the proofs of the results stated in the main text. We organize this part as follows. In Appendix \ref{ap:Lie}, the necessary prerequisites from the theory of Lie algebras and Lie groups are  given. Then, in Appendix \ref{ap:LieEx}, we present the relation between considered classes of linear optical gates and the irreducible representations of particular (semi-)simple Lie groups. In Appendix \ref{ap:aux}, we collect the auxiliary  results needed in the proofs of the theorems and the computations in examples given in the main text. In Appendix \ref{ap:examples} we proceed with the the computations related to the examples. Finally, in Appendix \ref{ap:mainTheorems} we give proofs of the main technical theorems.

Before we continue let us us first introduce the notation that will be  used in latter parts of the Appendix.

\begin{table}[h]
\begin{centering}
\begin{tabular}{|c|c|}
\hline
\textbf{Symbol} & \textbf{Explanation}\tabularnewline
\hline
$d$ & dimension of the local Hilbert space (number of modes)\tabularnewline
\hline
$\Hbos$ & Hilbert space of $N$ bosons in $d$ modes  \tabularnewline
\hline
$\Hbos$ & Hilbert space of $N$ fermions in $d$ modes  \tabularnewline
\hline
$\Hfock$ &  fermionic Fock space for  $d$  modes  \tabularnewline
\hline
$\Hfree$ &  positive-parity subspace of fermionic Fock space for  $d$ modes  \tabularnewline
\hline
$\LOB$ &  passive bosonic linear optics acting on $N$ bosons in  $d$ modes  \tabularnewline
\hline
$\LOF$ &  passive fermionic linear optics acting on $N$ fermions in  $d$ modes  \tabularnewline
\hline
$\FLO$ &  active fermionic linear optics acting on $\Hfree$ for $d$ modes \tabularnewline
\hline
$\SU(\H)$ &  special unitary group on Hilbert space $\H$  \tabularnewline
\hline
$\SO(\H)$ &  special orthogonal group on Hilbert space $\H$  \tabularnewline
\hline
$\USP(\H)$ &  unitary symplectic group on Hilbert space $\H$  \tabularnewline
\hline
$\SU(d)$ &  special unitary group on $\C^d$  \tabularnewline
\hline
$\Spin(2d)$ &  Spinor group of  $\R^{2d}$  \tabularnewline
\hline
$\Lie(K)$  &  Lie algebra of a Lie group $K$  \tabularnewline
\hline
$\Pi,\ \pi$ &  representations of a Lie group and  Lie algebra respectively   \tabularnewline
\hline
$\I$ & Identity operator on the relevant Hilbert space
\tabularnewline
\hline
$\T(\H)$ & Unitary gates proportional to identity on Hilbert space $\H$
\tabularnewline
\hline
$f_i\ ,\ f_{i}^\dagger$ &  fermionic anihilation and creation operators   \tabularnewline
\hline
$m_i$ &  Majorana-fermion operator   \tabularnewline
\hline
$Q$ &  fermionic parity operator acting in $\Hfock$   \tabularnewline
\hline
$\ket{n_1,n_2,\ldots,n_d}$ &  fermionic Fock state with occupation number $n_k$ in mode $k$   \tabularnewline
\hline
$[d]$ &  $d$ -element set $\lbrace1,2,\ldots,d\rbrace$  \tabularnewline
\hline

\end{tabular}

\par\end{centering}

\caption{Notation used throughout the paper.}
\end{table}

\section{Lie groups and Lie algebras}\label{ap:Lie}

Here we summarize some facts from the representation theory of Lie groups and Lie algebras that are used in the proofs of our results. This topic is a broad and fascinating one - we refer the reader to the textbooks \cite{HallGroups,FultonHarris,Humphreys1972} for a comprehensive introduction to the subject. Readers familiar with representation theory of Lie groups and Lie algebras can safely skip this part.

In this work we use extensively  \emph{compact} Lie groups, i.e., groups that are compact differential manifolds such that group multiplication is compatible with the differential structure. A unitary representation of a group Lie group $K$ in a Hilbert space $\H$ is a smooth mapping
\begin{equation}
\Pi:K\ni k \mapsto \Pi(k)\in \U(\H)\ \  \text{satisfying}\ \Pi(k_1 k_2)= \Pi(k_1)\Pi(k_2)\ \ \text{for all}\ k_1,k_2 \in K\ .
\end{equation}    
Compact Lie groups admit faithful (one-to-one) unitary  representations on finite-dimensional Hilbert spaces and therefore can be always modelled  as matrix subgroups of the unitary group $\U(\H_K)$, for a suitable Hilbert space $\H_K$. After fixing the  model of a Lie group $K$  we can define its Lie algebra as a real vector space of Hermitian operators that after exponentiation give elements from $K$,
\begin{equation}
\Lie\left(K\right)\coloneqq\SET{X\in\Herm\left(\H_K\right)}{\exp\left(\ii X\right)\in K} \ .
\end{equation} 
\begin{rem*}
Note that we decided to treat elements of the Lie algebra as Hermitian operators. Therefore for $X,Y\in \Lie(K)$ we have $\ii[X,Y]\in \Lie(K)$, instead of the property $[X,Y]\in\Lie(K)$, which occurs when elements of Lie algebra are taken to be skew-Hermitian. 
\end{rem*}

We say that a linear map $\pi:\Lie(K)\rightarrow\Herm(\H)$ is a a representation of  $\Lie(K)$ if and only if $\pi([X,Y])=[\pi(X),\pi(Y)]$, for all $X,Y\in\Lie(K)$.   A unitary representation $\Pi$ of a  Lie group $K$ in $\H$ induces the representation $\Pi_\ast$ of its Lie algebra $\Lie(K)$, by the map
\begin{equation}
\Pi_\ast: \Lie(K) \in X\mapsto \pi(X)\in \Herm(H)\ ,\ \text{defined by}\ \  \Pi_\ast(X)\DEF \left.\frac{d}{dt}\right|_{t=0} \Pi\left(\exp(\ii t X)\right)\ .
\end{equation}

A unitary representation $\Pi$ of $K$ in $\H$ we can promoted it to a unitary representation (usually dedonetd by $\Pi\ot\Pi$) of $K$  on $\H\otimes \H$ by  setting $\left[\Pi\ot\Pi\right](k)\coloneqq \Pi(k) \ot \Pi (k)$, for $k\in K$. Likewise, one can consider a ''tensor-square'' of the representation of the Lie algebra $\pi$ by defining $\left[\pi\ot\pi\right](X)\coloneqq\pi(X)\ot\I +\I\ot\pi(X)$, for $X\in\Lie(K)$.

A representation $\Pi$ of $K$ is called \emph{irreducible} if and only if there exist no proper (different thank ${0}$ or $\H$) subspace of $\W$ of $\H$  which is invariant under the action of $\Pi$ i.e. $\Pi(k) \W \subset \W$ for all $k\in K$. Otherwise, a representation $\Pi$ is called reducible. For compact Lie groups  any reducible representation decomposes onto a direct sum of irreducible representations i.e.
\begin{equation}
\H=\oplus_i \H_i \ , \Pi=\oplus_i \Pi_i\ ,
\end{equation}
where $\Pi_i$ is an irreducible unitary representation of $K$ in $\H_i$.  Analogously, a representation of a Lie algebra is called irreducible if and only if there is no proper invariant subspace $\W$ of the Hilbert space $\H$ such that $\pi(X) \W \subset \W$, for all $X\in\Lie(K)$. An induced representation of a Lie algebra $\pi$ is irreducible if and only if the representation of  of the Lie group $K$ is irreducible. Tensor-square representations of irreducible representations (of Lie groups or Lie algebras) are in general reducible.

Finally, we define the concept of \emph{simple} Lie algebra. We say that a Lie algebra $\Lie(K)$ is simple if and only if there exist no proper invariant subspace $\mathcal{I}$ of $\Lie(K)$ which satisfies
\begin{equation}
\ii[X,\mathcal{I}] \subset \mathcal{I}\ \  \text{for all}\ \ X\in\Lie(K)\ .
\end{equation}
A Lie algebra is called \emph{semisimple} if and only if it a direct sum of simple Lie algebras.

\section{Relevant Lie groups and Lie algebras}\label{ap:LieEx}

Here we give the definitions of classical simple Lie groups and algebras that are relevant for our considerations. Then we give the relation between considered classes of linear optical gates and the representations of these groups or corresponding Lie algebras.

\subsection*{Classical groups and Lie algebras} 
\noindent
The special unitary group $\SU(d)$ consists of unitary matrices on $\C^d$ of determinant one,
  
  \begin{equation}
\SU(d)=\SET{U\in\M_{d\times d}(\C)}{U U^\dagger =\I, \mathrm{det}(U)=1} \ .
\end{equation}
Its Lie algebra, denoted by $\su (d)$ consists of traceless Hermitian matrices on $\C^d$,
\begin{equation}
\su(d)=\SET{X\in\M_{d\times d}(\C) }{X =X^\dagger, \tr (X)=0}  \ .
\end{equation}
In the similar way we define special orthogonal group $\SO(d)$ and its Lie algebra $\so(d)$
  \begin{gather}
\SO(d)=\SET{O\in\M_{d\times d}(\R)}{O O^T =\I, \mathrm{det}(O)=1} \ , \\
\so(d)=\SET{X\in \ii \M_{d\times d}(\R) }{X =- X^T, \tr (X)=0}  \ .
\end{gather}
Lastly, in order to define unitary symplectic group $\USP(d)$ and its Lie algebra $\usp(d)$ we need to introduce the matrix $\J$ 
\begin{equation}
\J=\bigoplus_{i=1}^d \left(\begin{array}{cc}
0 & -1 \\ 
1 & 0
\end{array} \right)  \ ,
\end{equation}
which induces the symplectic form $B_a$ of $\C^{2d}$ by the expression
\begin{equation}
B_a (\ket{v},\ket{w}) =v^T \J w \ , 
\end{equation}
where $w,v$ in the right-hand side of above expression are column vectors consisting of components of vectors $\ket{v},\ket{w}\in\C^{2d}$ in the standard basis. Unitary symplectic group $\USP(d)$ consists of unitary matrices in $\C^{2d}$ that preserve this form,
\begin{gather}
\USP(2d)=\SET{O\in\M_{2d\times 2d}(\C)}{U U^{\dagger} =\I, U^T \J U = \J} \ , \\
\usp(2d)=\SET{X\in\M_{2d\times 2d}(\C) }{X =X^\dagger, X^T \J + \J X =0 }  \ .
\end{gather}
In what follows we will adopt the notation: $\SU(\H)$, $\su(\H)$, $\SO(\H)$, $\so(\H)$, $\USP(\H)$ and $\usp(\H)$, when we talk about the appropriate groups or Lie algebras on the abstract (finite-dimensional) Hilbert space $\H$.

\subsection*{Linear optical groups and Lie algebras}
\noindent
Lie algebras of linear optical groups discussed in the main text  are formally given by   
\begin{itemize}

\item for passive bosonic linear optics $\LOB$
\begin{equation} \label{eq:LieBOS}
	\Lie\left(\LOB \right)=\SET{ \left. \left(h\ot\I \ot\ldots \ot \I +
	\ldots + \I \ot \ldots \ot \I\ot h\right)\right|_{\Hbos}}{
	\ h\in\Herm\left(\C^d\right) }\ ; 
	\end{equation}
	\item for passive fermionic linear optics $\LOF$
		\begin{equation}
\Lie\left(\LOF \right)=\SET{ \left. \left(h\ot\I \ot\ldots \ot \I +
	\ldots + \I \ot \ldots \ot \I\ot h\right)\right|_{\Hfer}}{
	\ h\in\Herm\left(\C^d\right) }e \ ; 
	\end{equation}	
	\item for active fermionic linear optics $\FLO$	
		\begin{equation}
	\Lie\left(\FLO \right)=\SET{\left. \ii\left(\sum_{kl}h_{kl} m_k m_l +\alpha Q \right)\right|_{\Hfree}}{ h_{kl}=-h_{lk}\ ,\ h_{kl}\in\R,\ \alpha\in\R}
	\  .
	\end{equation}

We have included the term proportional to $\left. Q \right|_{\Hfree}$ as it generates unitary gates proportional to identity on $\Hfree$.\noindent
\end{itemize}

\noindent
The above Lie algebras after removing the center (consisting of hermitian operators proportional to identity on the appropriate Hilbert space) give the irreducible representations of simple Lie algebras. These Lie algebra representations are related to the representations of simply-connected compact simple Lie groups that we describe below (see Chapter 2 of \cite{Oszmaniec2014c}  for a detailed explanation of how these Lie algebras fit into a general picture of representation theory).

 For passive bosonic linear optics $\LOB$  acting on the space of $N$ bosons in $d$ modes, $\Hbos$, we have the associated representation of $\SU(d)$  which we denote by $\Pi_b$. we have
\begin{equation}
\SU(d)\ni U \mapsto \Pi_b(U)= \left.U^{\ot N}\right|_{\Hbos}\in\U(\Hbos) \ .
\end{equation}
The induced representation of Lie algebra $\mathfrak{su}(d)$ is denoted by $\pi_b$ and is given by
\begin{equation}
\su(d)\ni X \longmapsto \pi_b(X) = \left. \left(X\ot\I \ot\ldots \ot \I +
	\ldots + \I \ot \ldots \ot \I\ot X\right)\right|_{\Hbos} \ .
\end{equation}

\begin{rem}
For $d=2$ and $N$ particles the representation $\Pi_b$ corresponds, in the standard physics notation,  to the representation of $\SU(2)$ with the total spin $j=\frac{N}{2}$.
\end{rem}

Analogously, for passive fermionic linear optics $\LOF$  acting on the space of $N$ fermions in $d$ modes, $\Hfer$, we have the  representation of $\SU(d)$  which we denote by $\Pi_f$. It is defined by:
\begin{equation}
\SU(d)\ni U \mapsto \Pi_f(U)= \left.U^{\ot N}\right|_{\Hfer} \in\U(\Hfer)\ ,
\end{equation}
and the induced representation of Lie algebra $\mathfrak{su}(d)$ is denoted by $\pi_f$. It follows that
\begin{equation}
\su(d)\ni X \longmapsto \pi_f(X) = \left. \left(X\ot\I \ot\ldots \ot \I +
	\ldots + \I \ot \ldots \I\ot X\right)\right|_{\Hfer} \ .
\end{equation} 
For the active fermionic linear optics we have an irreducible spinor representation of the Lie algebra $\so(2d)$ in $\Hfree$  denoted by $\pi_{\FLO}^+$. We start by describing a representation $\pi_{\FLO}$ that acts in the full Fock space $\Hfock$. We define this representation by its action on the orthogonal (with respect to Hilbert-Schmidt inner product) basis of $\su(2d)$ of the form (see Chapeter 2 of \cite{Oszmaniec2014c} for details)
\begin{equation}\label{eq:BASISso2D}
E_{kl}=\ii (\kb{k}{l} - \kb{l}{k})\ ,\ 1\leq k < l \leq 2d \ .  
\end{equation}
We have
\begin{equation}\label{eq:SpinorREP}
\pi_{\FLO}(E_{kl})= \frac{\ii}{2}\ m_k m_l\ ,\ 1\leq k < l \leq 2d \ . 
\end{equation}
This representation  decomposes onto two irreducible components supported on positive and negative parity subspaces of the Fock space. We denote the restriction of $\pi_{\FLO}$ to $\Hfree$ by $\pi_{\FLO}^+$. Formally, we have
\begin{equation}
\pi_{\FLO}^+\left(\so(2d)\right) =  \SET{\left. \ii\left(\sum_{kl}h_{kl} m_k m_l \right)\right|_{\Hfree}}{ h_{kl}=-h_{lk}\ ,\ h_{kl}\in\R,\ \alpha\in\R} \ .
\end{equation}
The corresponding representations of the group $\Spin(2d)$ (double-cover of  $\SO(2d)$) in $\Hfree$ and $\Hfock$ are denoted by $\Pi_{\FLO}^+$ and $\Pi_{\FLO}$ respectively.

\section{Auxiliary technical results}\label{ap:aux}

In this part we state and prove a number of auxiliary results that are necessary for the computations concerning the examples (see Appendix \ref{ap:examples}) and in the proofs of the main results (see Appendix \ref{ap:mainTheorems}).

\begin{lem}[Projection for the singlet subspace for the doubled bosonic representation of $\SU(2)$)\label{Projbosons}] 	Let $d=2$ and let $\Pi_b \ot \Pi_b$ be the representation of $\SU(2)$ in $\Hbos\ot\Hbos$ induced from representation $\Pi_b$ on $\Hbos$.  Then, in the decomposition of $\Hbos\otimes\Hbos$ onto irreducible components there always appears a a single trivial representation of $\SU(2)$. Moreover, the projection onto this representation is given by
	\begin{equation}\label{eq:projBOS}
	\P_{b}=\frac{1}{N+1}\L_b=\frac{1}{N+1} \kb{\Psi_b}{\Psi_b}  \ ,
	\end{equation}
	where 
	\begin{equation}\label{eq:PsiBapp}
	\ket{\Psi_b}=\sum_{k=0}^N (-1)^k  \ket{D_k}\ket{D_{N-k}} \ .
	\end{equation}
\end{lem}
\begin{proof}
	The fact that in the decomposition of $\Hbos\otimes\Hbos$ there is always a single trivial representation follows from the standard rules of addition of angular momentum \cite{Sakurai2011}.  The formula \eqref{eq:projBOS} can be derived using techniques involving Young diagrams. However, checking whether $\ket{\Psi_b}$ actually belongs to the trivial representation of $\SU(2)$ can be done easily by inspecting the action of $\Lie(\SU(2))$ on it.  This action comes from the fact that just like $\SU(2)$, its Lie algebra is represented in $\Hbos$ via standard angular momentum representation $\pi_b$.  The standard basis of $\SU(2)$ is represented in $\Hbos\ot\Hbos$ via 
	\begin{equation}
	\pi_b \ot \pi_b(\sigma_x)=J_x\ot\I+\I\ot J_x \ , \pi_b \ot\pi_b (\sigma_y)=J_y\ot\I+\I\ot J_y \ ,\ \pi_b \ot\pi_b (\sigma_z)=J_z\ot\I+\I\ot J_z \ .
	\end{equation}
	where $J_{x,y,z}$ denote the standard angular momentum operators (for spin $s=\frac{N}{2}$). 
	Now it suffices to show that 
	\begin{equation}\label{eq:BOScond}
		\left[\pi_b \ot \pi_b(\sigma_\alpha)\right]\ket{\Psi_b}=0\ ,
	\end{equation}
	for $\alpha=x,y,z$. Equation \eqref{eq:BOScond} can be verified easily using the standard algebraic properties of operators $J_{x,y,z}$ (recall that the Dicke basis is exactly the standard  "angular momentum" basis in which $J_z$ is diagonal.
\end{proof}

\begin{lem}(Projection for the singlet subspace for the doubled fermionic half-filling representation of $\SU(d)$) \label{lem:Projfermions} 	
Let $d=2N$ and let  and let $\Pi_f \ot \Pi_f$ be the representation of $\SU(d)$ in $\Hfer \ot \Hfer$ induced from representation $\Pi_f$ on $\Hfer$. In the decomposition of $\Hfer\otimes\Hfer$ onto irreducible components there always appears a a single trivial representation of $\SU(d)$. Moreover, the projection onto this representation is given by
\begin{equation}\label{eq:projFER}
\P_{f}=\L_f= \kb{\Psi_f}{\Psi_f}  \ ,
\end{equation}
where 
\begin{equation}\label{eq:PsiFER1}
\ket{\Psi_f}=\ket{1}\wedge\ket{2}\wedge\ldots\wedge\ket{2N} \ .
\end{equation}
\end{lem}

\begin{proof}
The decomposition of the tensor square $\Hfer\ot \Hfer$ onto irreducible representations of $\SU(d)$ (for arbitrary number of particles $N$) can be found in \cite{symmetricFunctions1998} on page 331. From the results given there it follows that for $d=2N$  in the decomposition  $\Hfer\ot \Hfer$ there exist only one trivial one dimensional representation $\H^0_f$. In other words in the case of half filling we have
\begin{equation}
\Hfer\ot \Hfer \approx \tilde{\H}_f \oplus \H^0_f \ ,
\end{equation}    
where $\tilde{\H}_f$ the orthogonal complement of $\H^0_f$ in $\Hfer\ot \Hfer$. Since $\H^0_f$ is one dimensional, the projection onto this space is specified by a single vector belonging to it. To find  this vector  we first  embed $\Hfer$ into the tensor power representation of $\SU(d)$, that is we note that
\begin{equation}
\Hfer = \bigwedge^N \left(\C^d\right) \subset \left(\C^d\right)^{\ot N} \ ,
\end{equation}\
and
\begin{equation} 
\Pi_f (U) = \left.U^{\ot N}\right|_{\Hfer}\ \text{for}\ U\in \SU(d) \ . 
\end{equation}
Under this inclusion we have $\Hfer\ot \Hfer \subset  \left(\C^d\right)^{\otimes 2N}$ and consequently
\begin{equation}\label{eq:tensorPOWERfer}
\Pi_f \ot \Pi_f = \left.U^{\ot 2 N}\right|_{\Hfer \ot \Hfer} \ .
\end{equation}
Now, $\H^0_f$ can be treated as a subspace of $\left(\C^{2N}\right)^{\otimes 2N}$ (recall that we consider the case $2N=d$). Moreover, we see that for the case of half-filling the vector  $\ket{\Psi_f}$ given in \eqref{eq:PsiFER1}  belongs to $\Hfer\ot \Hfer$. This follows form the fact that it is a totally antisymmetric vector on $\left(\C^{2N}\right)^{\otimes 2N}$. What is more, $\ket{\Psi_f}\in \H^0_f$ which is a result of \eqref{eq:tensorPOWERfer} and the standard properties of  the wedge product: For $U\in\SU(d)$ we have $\left[\Pi_f \ot \Pi_f (U)\right] \ket{\Psi_f}=U^{\ot 2N}  \ket{\Psi_f} = \mathrm{det}(U)  \ket{\Psi_f} =  \ket{\Psi_f}$, 
where in the last equality we have used the fact that  $\mathrm{det}(U)=1$ for $U\in\SU(d)$.
\end{proof}
\noindent
Let us introduce some auxiliary notation. For a subset $X$ of $d$-element set $[d]$ we define its indicator function by
\begin{equation}\label{eq:indicator}
X(i) \DEF \begin{cases}
    1\ ,& \text{if } i\in X\\
    0\ ,              & \text{if } i \notin X\end{cases}\ .
\end{equation}
We denote by  $\bar{X}$ the complement of $X$ in $[d]$.  Moreover, for the case of half-filling ($2N=d$) for a set $X=\lbrace x_1, x_2 ,\ldots, x_{N}  \rbrace $  we define $\sgn(X)=\sgn(\omega_X)$, where $\sgn(\omega_X)$ is the sign of the permutation $\omega_X: [2N]\rightarrow [2N]$ such that
\begin{itemize}
\item $\omega_X(i)$ ($i=1,\ldots,N$) are non-decreasingly ordered elements of $X$;
  \item $\omega_X(i)$ ($i=N+1,\ldots,2N$) are non-decreasingly ordered elements of $\bar{X}$. 
\end{itemize}

\begin{lem}[Convenient form of the projection onto  the trivial representation of $\SU(d)$ in $\Hfer \ot \Hfer$ for $2N=d$] \label{lem:ProjPassiveDiff}  For the case of half-filling ($N=2d$) the invariant vector $\ket{\Psi_f}\in \Hfer \ot \Hfer$ can be written in the following way

\begin{equation}\label{eq:LferVEC}
\ket{\Psi_f}=\frac{1}{\sqrt{\binom{2N}{N} }}\sum_{X\subset[2N],\  |X|=N} \sgn(X) \ket{X} \ot \ket{\bar{X}}\ ,
\end{equation}
where $\ket{X}=\ket{X(1),X(2),\ldots,X(2N)}$.
\end{lem}
\begin{proof}
Using Lemma \ref{lem:Projfermions} we know that $\ket{\Psi_f}= \ket{1}\wedge\ket{2}\wedge\ldots\wedge\ket{2N}$.  Expanding $\ket{1}\wedge\ket{2}\wedge\ldots\wedge\ket{2N}$ in the first quantisation picture gives
\begin{equation}\label{eq:expansionSLATER}
\ket{\Psi_f} = \frac{1}{\sqrt{(2N)!}} \sum_{\sigma\in P([2n])    } \sgn(\sigma) \ket{\sigma(1)}\ot \ket{\sigma(2)}\ot \ldots \ot \ket{\sigma(2N)} \ ,
\end{equation}
where $P([2N]$ denotes the permutation group of the set $[2N]$ and $\lbrace \ket{i} \rbrace_{i=1}^d $ is the fixed orthonormal basis of $\C^d$. Let $\P^N_{\mathrm{as}}:(\C^{2N})^{\ot N} \rightarrow (\C^{2N})^{\ot N}$ be the projection onto the totally antisymmetric subspace of $(\C^{2N})^{\ot N}$. From the asymmetry of  $\ket{1}\wedge\ket{2}\wedge\ldots\wedge\ket{2N}$ we have $\P^N_{\mathrm{as}} \ot \P^N_{\mathrm{as}} \ket{\Psi_f} = \ket{\Psi_f}$. Using this and \eqref{eq:expansionSLATER} we obtain
\begin{equation}\label{eq:expansionSLATER2}
\ket{\Psi_f} = \frac{1}{\sqrt{(2N)!}} \sum_{\sigma\in P([2n])} \sgn(\sigma) \P^N_{\mathrm{as}} \left( \ket{\sigma(1)}\ot \ldots \ot \ket{\sigma(N)} \right) \ot  
\P^N_{\mathrm{as}} \left( \ket{\sigma(N+1)} \ot \ldots \ot \ket{\sigma(2N)} \right)\ .
\end{equation}
In order to simplify this expression we note we have a decomposition 
\begin{equation}\label{eq:permDECOMP}
\sigma = \tau_I(\sigma) \circ \tau_{II}(\sigma)  \circ \omega_{\sigma(\lbrace{1,\ldots,N\rbrace)}}\ ,
\end{equation} 
where
\begin{itemize}
\item  $\omega_{\sigma([N])}$ has been defined above the formulation of the Lemma ($\sigma([N])$ is a particular subset of $[2N]$) ;
\item $\tau_I(\sigma)\in P([2N])$ is the unique permutation that maps the non-increasingly ordered elements  of $\sigma([N])$ to the ordered tuple $(\sigma(1),\sigma(2), \ldots, \sigma(N)$ and acts as identity on $\sigma(\lbrace{N+1,\ldots,2N\rbrace)}$; 
\item $\tau_I(\sigma)\in P([2N])$ is the unique permutation that maps the non-increasingly ordered elements  of $\sigma(\lbrace{N+1,\ldots,2N\rbrace})$ to the ordered tuple $(\sigma(N+1),\sigma(N+2), \ldots, \sigma(2N))$ and acts  as identity on $\sigma(\lbrace{1,\ldots,N\rbrace)}$ .
\end{itemize}
Furthermore, we note that due to the antisymmetry of $\P^{N}_{\mathrm{as}}$ 
\begin{eqnarray}
 \P^N_{\mathrm{as}} \left( \ket{\sigma(1)}\ot \ldots \ot \ket{\sigma(N)} \right) &=&  \sgn(\tau_{I}(\sigma))  \P^N_{\mathrm{as}} \left( \ket{\omega_{\sigma([N])}(1)}\ot \ldots \ot \ket{\omega_{\sigma([N])}(N)} \right) \label{eq:semiSLAT1} \\
  \P^N_{\mathrm{as}} \left( \ket{\sigma(N+1)}\ot \ldots \ot \ket{\sigma(2N)} \right) &=& \sgn(\tau_{II}(\sigma))  \P^N_{\mathrm{as}} \left( \ket{\omega_{\sigma([N])}(N+1)}\ot \ldots \ot \ket{\omega_{\sigma([N])}(2N)} \right) \label{eq:semiSLAT2}  \ .
\end{eqnarray}
However, from the definition of the Slater determinant we have 
\begin{equation}\label{eq:slat1}
\P^N_{\mathrm{as}} \left( \ket{\omega_{\sigma([N])}(1)}\ot \ldots \ot \ket{\omega_{\sigma([N])}(N)} \right) = \frac{1}{\sqrt{N!}}  \ket{\sigma([N])}\ ,
\end{equation}
and similarly for the complement set
\begin{equation}\label{eq:slat2}
\P^N_{\mathrm{as}} \left( \ket{\omega_{\sigma([N])}(N+1)}\ot \ldots \ot \ket{\omega_{\sigma([N])}(2N)} \right) = \frac{1}{\sqrt{N!}}  \ket{\sigma(\bar{[N]})}\ .
\end{equation}
Combining together the above identities in the expression \eqref{eq:expansionSLATER2} we obtain
\begin{equation} 
\ket{\Psi_f}=\frac{1}{N! \sqrt{(2N)!}} \sum_{\sigma\in P([2n])    } \sgn(\sigma) \sgn(\tau_{I}(\sigma)) \sgn(\tau_{II}(\sigma))  \ket{\sigma([N]) }\ot \ket{ \sigma(\bar{[N]})}\ .
\end{equation}
The above can be further simplified by exploiting \eqref{eq:permDECOMP} and the homomorphism property of the $\sgn$ function,
\begin{eqnarray}
\ket{\Psi_f} & = &\frac{1}{N! \sqrt{(2N)!}} \sum_{\sigma\in P([2n])    } \sgn(\sigma([N])) \ket{\sigma([N]) }\ot \ket{ \sigma(\bar{[N]})}  \\
 & = & \frac{1}{\sqrt{\binom{2N}{N}}} \sum_{X \subset [2N], |X|= N} \sgn(X) \ket{X }\ot \ket{\bar{X}} \label{eq:finalSLATER} \ ,
\end{eqnarray}
where in the last equality we have a changed a sum over permutations to the sum over  $N$-element subsets of $[2N]$, which resulted in the binomial coefficient in front of the sum in \eqref{eq:finalSLATER}.
 
\end{proof}

\begin{rem}
It can be showed that the function $\sgn(X)$ defined above Lemma \ref{lem:ProjPassiveDiff} is given explicitly via
\begin{equation}\label{eq:EXPLICITsgn}
\sgn(X)= (-1)^{\sum_{i=1}^N X(2i) }\ \text{where } X\subset[2N],\ |X|=N\ .  
\end{equation}
Using that $\sgn(X) \sgn(\bar{X}) = (-1)^N$, we obtain 
\begin{equation}
\ket{\Psi_f}\in \mathrm{Sym}^2(\Hfer)\ \text{ for $N=2m$ and }\ \ket{\Psi_f}\in \bigwedge^2(\Hfer) \ \text{ for  $N\neq 2m$\ .} 
\end{equation}
\end{rem}

\begin{lem}[Projection onto the trivial representation of  $\Spin(2d)$ in $\Hfree \ot \Hfree$ for even $d$] \label{ProjFree}	 
Assume that  $d=2k$ and let  and let $\Pi_\FLO \ot \Pi_\FLO$ be the representation of $\Spin(2d)$ in $\Hfree \ot \Hfree$ induced from representation $\Pi_\FLO$ on $\Hfree$. In the decomposition of $\Hfree\otimes\Hfree$ onto irreducible components there always appears a a single trivial representation of $\Spin(2d)$. Moreover, the projection onto this representation is given by
\begin{equation}\label{eq:projFLO}
\P^+_{\FLO}=\P^+ \ot \P^+ \L_\FLO \P^+ \ot \P^+\ ,
\end{equation}
where 
\begin{equation}
\L_\FLO=  \frac{1}{2^{d(2d-1)}}  \prod_{1\leq k<l\leq 2d} \left(\I\ot \I + m_k m_l\ot m_k m_l \right) \ 
\end{equation}
and $\P^+ =\frac{1}{2}(\I +Q)$ is the orthonormal projection onto $\Hfree\subset \H_\mathrm{Fock}$.
\end{lem}

\begin{proof}
It will be convenient for us to work first in the full fermionic Fock space $\Hfock$, that carries a (reducible) representation of $\Spin(2d)$ denoted by $\Pi_{\FLO}$ (see the previous section of the Appendix). We can promote this representation to the representation $\Pi_{\FLO}\ot \Pi_{\FLO}$ of $\Spin(2d)$ in $\Hfock\ot\Hfock$. We can now identify the trivial representation of  $\Spin(2d)$ (or equivalently $\so(2d)$) in $\Hfock\ot\Hfock$ with the zero eigenspace of the second order Casimir \cite{FultonHarris,Oszmaniec2014c} of $\so(2d)$ represented in $\Hfock\ot\Hfock$. The Casimir operator, denoted by $\mathcal{C}_2$ is given by 
	\begin{align}
	\mathcal{C}_2 &= \sum_{1\leq k<l\leq 2d } \left( \pi_{\FLO}(E_{kl}) \ot \I +\I\ot \pi_{\FLO}(E_{kl})\right)^2\ \label{EQ} \\
	&= \frac{1}{2} \sum_{1\leq k<l\leq 2d } \left(\I\ot \I -  m_k m_l \ot m_k m_l \right) \ \label{eq:casFLO}, 
	\end{align} 
	where in the computation above we have used \eqref{eq:SpinorREP} and the identity $(\ii m_k m_l)^2= \I$, valid for all $k\neq l$. From Eq.  \eqref{eq:casFLO} we see that $\mathcal{C}_2$ is a sum of $2d(2d-1)/2$ commuting positive operators  $F_{kl}= \I\ot \I -  m_k m_l \ot m_k m_l$ having eigenvalues $0$ and $1$. The projection onto the zero eigenspace of $F_{kl}$ in  is given by
	\begin{equation}
	\P_{kl}=\frac{1}{2}\left(\I\ot \I +  m_k m_l \ot m_k m_l \right)\ .	\end{equation}
	Finally, the projection onto the zero eigenspace of $ \mathcal{C}_2$ is the product of all these projections 
	\begin{equation}
	\L_\FLO= \frac{1}{2^{d(2d-1)}}  \prod_{1\leq k<l\leq 2d} \left(\I\ot \I + m_k m_l\ot m_k m_l \right)\ .
	\end{equation}
By restricting $\L_{FLO}$ to the subspace $\Hfree \ot \Hfree \subset \Hfock \ot \Hfock$ we finally get that $\P^{+}_{\FLO}$ is the projection onto a trivial representation of $\Spin(2d)$ in $\Hfree \ot \Hfree$. Note however that 
	\begin{equation}\label{eq:projIneqFLO}
	\L_\FLO \leq \mathbb{A}_\FLO=\frac{1}{2^d}\prod_{i=1}^d \left(\I\ot \I + m_{2i-1} m_{2i}\ot  m_{2i-1} m_{2i}\right) \ .
	\end{equation}
It is easy to see that the operator $\mathbb{A}_\FLO$ has support on $\Hfock^+\otimes \Hfock^- \oplus \Hfock^-\otimes \Hfock^+$ for odd $d$ and on  $\Hfock^+\otimes \Hfock^+ \oplus \Hfock^-\otimes \Hfock^-$ . We thus see that $\P^{+}_\FLO =0$ unless $d$ is even. The fact that $\P^{+}_\FLO\neq 0$ for even $d$  follows form the fact that in this case positive-parity spinor representations of $\Spin(2d)$ are self-dual - see the remark below Lemma \ref{ProjActiveDiff} below.  
\end{proof}

\begin{lem}[Another form of the projection onto the trivial representation of $\Spin(2d)$ in $\Hfree \ot \Hfree$, for even $d$] \label{ProjActiveDiff}

The projection $\P_{\FLO}^+$ onto the trivial representation of $\Spin(2d)$ in $\Hfree \ot \Hfree$ can be expressed via a simpler expression
\begin{equation}\label{eq:FLOprojALT}
\P_{\FLO}^+ = \P^+ \ot \P^+ \L'_\FLO \P^+ \ot \P^+\ ,
\end{equation}
with 
\begin{equation}\label{eq:Lalt}
 \L'_\FLO=  \frac{1}{2^{2d}}  \left(\prod_{i=1}^{d} \left(\I\ot \I + m_{2i-1} m_{2i}\ot  m_{2i-1} m_{2i}\right) \right)  \left(\prod_{j=1}^{d} \left(\I\ot \I + m_{2j} m_{2j+1}\ot  m_{2j} m_{2j+1}\right) \right) \ .
\end{equation}	
In the second bracket in above expression we used the convention $m_{2d+1} \equiv m_1$.	
\end{lem}
\begin{proof}
We first note that the operators $Q_{k,l}=m_k m_l \ot m_k m_l$ (recall that $1\leq k<l\leq 2d$) can be always obtained as a product of ''nearest-neighbour'' operators $Q_{i,i+1}$ ($i=1,\ldots,2d-1$). Indeed, we have
\begin{equation}
Q_{k,l} = \prod_{i=k}^{l-1} Q_{i,i+1}
\end{equation} 
and therefore the joint +1 eigenspace of $Q_{k,l}$ ($1\leq k<l\leq 2d$) is exactly the joint +1 eigenspace of $Q_{i,i+1}$, for $i=1,\ldots,2d-1$. Consequently, we have $\L_{\FLO}=\L'_\FLO$ and equation \eqref{eq:FLOprojALT} follows (note that in \eqref{eq:Lalt} we have incorporated, for the sake of symmetry of the expression, a redundant term projecting onto the +1 eigenspace of $Q_{1,2d}$).
\end{proof}
\begin{rem}
From expressions \eqref{eq:FLOprojALT} and \eqref{eq:Lalt} it follows that $\P_\FLO^+ \neq 0$ if and only if $d$ is even. To show this we first observe that $\L'_{\FLO}$ commutes with operators $Q\ot\I$ and $\I\ot Q$ and hence it commutes with the projection $\P^+ \ot \P^+$. Therefore in order to prove $\P^+ \ot \P^+   \L'_\FLO \P^+ \ot \P^+ $ it suffices to show $\tr( \P^+ \ot \P^+ \L'_\FLO) \neq 0$. To prove this we first expand  the products inside the brackets of \eqref{eq:Lalt} and obtain
\begin{equation}
\L_\FLO = \frac{1}{2^{2d}}\left(\I\ot\I + (-1)^d Q\ot Q + L_1   \right)  \left(\I\ot\I + (-1)^d Q\ot Q + L_2   \right)\ ,  
\end{equation}
where $L_1$ and $L_2$ contain sums of operators  $m_{X}\ot m_{X}$, with $m_X= \prod_{i\in X} m_i$ and $X$ is a proper subset of $[2d]=\lbrace 1,\ldots, 2d\rbrace$. Moreover, from the expansion of the brackets in   \eqref{eq:Lalt} we see that for every $m_X\ot m_X$ appearing in $L_1$ there exist no $m_Y\ot m_Y$ in $L_1$ such that $X\cup Y=[2d]$. From this discussion and the fact that products of Majorana fermion operators are traceless we finally obtain 
\begin{equation}
\tr\left(\P^+ \ot \P^+ \L'_\FLO\right)=\frac{1}{2^{2d}4}\tr\left[(\I\ot\I +Q\ot Q)(2\I\ot\I +2 (-1)^d Q\ot Q)\right]=\frac{1}{2}(1+(-1)^d)\ ,
\end{equation}      
where we have used the identities $Q^2=\I$, $\tr(\I\ot\I)=2^{2d}$ and $\tr(Q\ot Q)=0$. Analogous computations show that 
\begin{equation}
\tr\left(\P^- \ot \P^- \L'_\FLO\right)=\frac{1}{2}(1+(-1)^d)\ ,
\end{equation}
where $\P^-=\frac{1}{2}(\I-Q)$. Similarly, we have
\begin{equation}
\tr\left(\P^+ \ot \P^- \L'_\FLO\right)=\tr\left(\P^+ \ot \P^- \L'_\FLO\right)=\frac{1}{2}(1(-1)^{d}),
\end{equation}
and thus we conclude that the trivial representation appears in $\Hfree \ot \Hfock^-$ (or in $\Hfock^- \ot \Hfree $) only when $d$ is odd.    
\end{rem}

\noindent
We close the section of  the auxiliary results by giving the explicit form of vector $\ket{\Psi_\FLO}$.

\begin{lem}[Explicit form of the vector spanning the trivial representation of $\Spin(2d)$ in $\Hfree \ot \Hfree$ for even $d$.]\label{lem:ExplictitProjFLO} For a subset $X$ of $d$-element set $[d]$  we define $N(X)=\sum_{i\in X} i $. Under this notation we have
\begin{equation}\label{eq:invariantVectorFLO} 
\P^{+}_\FLO = \kb{\Psi_\FLO}{\Psi_\FLO}\ ,\ \text{with}\  \ket{\Psi_\FLO}=\frac{1}{\sqrt{2^{d-1}}}\sum_{X\subset[d], |X|=2k} (-1)^{N(X)} \ket{X}\ot \ket{\bar{X}} \ ,
\end{equation}
where $\ket{X}=\ket{X(1),X(2),\ldots,X(d)}$ and we have used the notation introduced above Lemma \ref{lem:ProjPassiveDiff}.
\end{lem}
\begin{proof}
From the relation $\L_\FLO \leq \mathbb{A}_\FLO$ (see Eq.\eqref{eq:projIneqFLO})  we can deduce the following observations
\begin{itemize}
\item [(i)] Operator $\P^{+}_\FLO$ is supported on the subspace $\Hfock \ot \Hfock$ spanned by linearly-independent vectors $\ket{X}\ot\ket{\bar{X}}$ for even-element subset $X\subset[d]$;
\item [(ii)] Operator  $\P^+_\FLO$ must necessary project onto one dimensional subspace (this follows form the basic character theory for the compact groups \cite{ZZ2015});
\item [(iii)] $\P^{+}_\FLO$ must be a projection onto a maximally entangled vector (denoted by $\ket{\Psi_{\FLO}}$) in the space $\Hfree \ot\Hfree$. If this was not the case then the invariant form $B(\ket{\psi},\ket{\psi})=\bra{\Psi}\ket{\psi}\ot\ket{\phi}$, for $\ket{\psi},\ket{\phi}\in\Hfree$ would have been degenerate which would contradict the irreducibility of the representation $\Pi^{+}_{\FLO}$ in $\Hfree$.
\end{itemize}
Combining these observations we get
\begin{equation}\label{eq:almostPROJ}
\ket{\Psi_\FLO}= \frac{1}{\sqrt{2^{d-1}}}\sum_{X\subset[d], |X|=2k} \exp(\ii \theta_X) \ket{X}\ot \ket{\bar{X}}\ . 
\end{equation}
The global phase of $\ket{\Psi_\FLO}$ can be set  arbitrary way. In what follows we set without the loss of generality  $\theta_\emptyset =0$. As we will see this choice fixes the values of all remaining phase factors to 
\begin{equation}
\exp(\theta_X)=(-1)^{N(X)}\ .
\end{equation}
 In order to prove this let us study how the unitary  operator 
 \begin{equation}
E_{ij}\coloneqq  m_{2i-1} m_{2j-1} \ot m_{2i-1} m_{2j-1}\ , \ i<j\ ,\ i,j\in[d]\ ,
 \end{equation}
 acts on the vector  $\ket{\emptyset}\ot \ket{[d]}$. Explicit computation gives
 \begin{equation}
 E_{ij} \ket{\emptyset}\ot \ket{[d]} = (-1)^{i+j}\ket{\lbrace i,j \rbrace } \ot \ket{[d]\setminus \lbrace i,j \rbrace } = (-1)^{N(\lbrace i,j \rbrace)}  \ket{\lbrace i,j \rbrace } \ot \ket{[d]\setminus \lbrace i,j \rbrace }\ .
 \end{equation}
 Defining analogously the operator $E_X$, where $X\subset[d]$ is even-element subset of $[d]$ we get
 \begin{equation}
  E_X \ket{\emptyset}\ot \ket{[d]} = (-1)^{N(X)}\ket{X}\ot \ket{\bar{X}}\ .
 \end{equation}
The crucial observation now is that operators $F_X$ are in  $\Pi^+_\FLO \ot \Pi^+_\FLO (\Spin(2d))$ and therefore 
\begin{equation}\label{eq:invFLO}
F_X \ket{\Psi_\FLO} = \ket{\Psi_\FLO}\ ,\ \text{for all } X\subset[d],\ |X|=2k\ . 
\end{equation}
Moreover, we note that 
\begin{equation}
F_X \ket{Y}\ot\ket{\bar{Y}} \propto \ket{Y+X}\ot\ket{[d]\setminus \lbrace{Y+X\rbrace}}\ ,\  X,Y\subset[d],\ |X|=2k, |Y|=2l \ ,
\end{equation}
where $X+Y=X\cup Y \setminus X\cap Y$ denotes the symmetric sum of $X$ and $Y$. Using the above relation together with \eqref{eq:invFLO}, assumption $\theta_\emptyset =0$,  and \eqref{eq:almostPROJ}  we finally obtain \eqref{eq:invariantVectorFLO}.
 
\end{proof}

\begin{rem}
Using the positive-parity constrain we get $N(\bar{X})+N(X)=(-1)^{d/2}$. Consequently we have
\begin{equation}
\ket{\Psi_\FLO}\in \mathrm{Sym}^2(\Hfree)\ \text{ for $d=4m$ and }\ \ket{\Psi_\FLO}\in \bigwedge^2(\Hfree) \ \text{ for  $d\neq 4m$\ .} 
\end{equation}
\end{rem}

\begin{rem}
It is possible to prove that for $d=4m$ we have
\begin{equation}
\P_{d/2}\ot \P_{d/2} \ket{\Psi_\FLO} \propto \ket{\Psi_{f}}\ ,  
\end{equation}
where $\ket{\Psi_f}$ is the passive-FLO invariant vector (see Eq. \eqref{eq:LferVEC}) and $\P_{d/2}:\Hfree\rightarrow \Hfree$ is the projection onto the half-filling subspace $\bigwedge^{d/2} \left( \C^d \right) \subset \Hfree$. We suspect that analogous relation holds also for the case of even $d$ which is not divisible by $4$. We leave proofs of these statements to the interested reader.
\end{rem}

\section{Detailed computations concerning examples from the main text}\label{ap:examples}

In this section we present explicit computations concerning examples given in the main text.

\subsection{Example 1} \label{EX1:proof}

The result given in Example \ref{ex:CROSSKERR} follows directly from Theorem \ref{thm:BOSGATES}, or more specifically from the situation described in the possiblity (c). It states that the group generated by passive bosonic linear optics and $V\in\U\left(\Hbos\right)$ is the full unitary group $\U\left(\Hbos\right)$ if and only if $\left[V\ot V,\L_b\right]\neq0$, where $\L_b$ is given by Eq. \eqref{eq:Lbos}. Let $V_t=\exp\left(-it\hat{n}_{a}\hat{n}_{b}\right)$. Straightforward computation gives 
\begin{equation}
	\left[V_{t}\ot V_{t},\L_b\right] =\sum_{k,l=0}^{N} f_{kl} \kb{D_{k}}{D_{l}}\ot\kb{D_{N-k}}{D_{N-l}}\ ,
\end{equation}
where
\begin{equation}\label{eq:condKERR}
f_{kl}=\left(-1\right)^{l+k}\e^{-2itk\left(N-k\right)}\left(\e^{2it\left[l\left(N-l\right)-k\left(N-k\right)\right]}-1\right)\ .
\end{equation}
From the above expressions it follows that $\left[V_{t}\ot V_{t},\L_b\right] =0$ if and only if $f_{kl}=0$ for all $k,l=0,\ldots,N$ which is equivalent to \eqref{eq:conditionKerr}. Solving \eqref{eq:condKERR} is an interesting
problem in itself. We limit ourselves to proving that for time  $t=\frac{\pi}{3}$ the gate $V_{t}$ promotes $\LOB$ to
universality. The reason for that is the following: number $3$ does
not simultaneously divide integers of the form $n_{l}=l\left(N-l\right)$ ($l=0,\ldots,N$).
To  see this we first observe that for $N=1,\ldots,5$ the sequence $n_{l}$ takes values 
\begin{itemize}
	\item $N=1:$ $n_l=(0,0)$;
	\item $N=2:$ $n_l=(0,1,0)$;
	\item $N=3:$ $n_l=(0,2,2,0)$;
	\item $N=4:$ $n_l=(0,3,4,3,0)$;
	\item $N=5:$ $n_l=(0,4,8,9,8,4,0)$.
\end{itemize}
For $N>5$ in the sequence $n_{l}$ we have numbers $\left(N-2\right)2$
and $\left(N-4\right)4$ which cannot be both simultaneously divisible
by $3$. Consequently, we have $\left[V_\frac{\pi}{3}\ot V_\frac{\pi}{3},\L_b\right]\neq 0$.

\subsection{Example 2}\label{EX2:proof}

In order to prove the content of the Example \ref{ex:BOSmid} we refer to the point (b) of Theorem \ref{thm:BOSHAM}. Form there it follows that $\langle \LOB, X_3 \rangle = G_b$ is and only if $[X_3 \ot \I + \I \ot X_3 , \L_b]=0$. Recall that $\L_b \propto \kb{\Psi_b}{\Psi_b}$, where $\ket{\Psi_b}=\sum_{k=0}^N (-1)^k  \ket{D_k}\ket{D_{N-k}}$. Using $X_3 =\hat{n}^3_a - \hat{n}^3_b$ we obtain
\begin{equation}
\left(X_3 \ot \I + \I \ot X_3\right) \ket{D_k}\ket{D_{N-k}} = \left[\left((k)^3 -(N-k)^3\right)+ \left((N-k)^3 -(k)^3 \right)\right] \ket{D_k}\ket{D_{N-k}}= 0\ .
\end{equation}  
Therefore $\left(X_3 \ot \I + \I \ot X_3\right)  \ket{\Psi_b}=0$ and consequently $[X_3 \ot \I + \I \ot X_3, \L_b]=0$. As a result we have $\langle \LOB, X_3 \rangle = G_b$. The specific cases (a) and (b) discussed in the Example \ref{ex:BOSmid} follow immediately from the discussion below Theorem \ref{thm:BOSGATES}.

\subsection{Example 3}\label{EX3:proof}

Any Hamiltonian $H_{\text{tm}}$ acting on $\Hfer$ that contains only two-mode terms has the following general form:
\begin{equation}
H_{\text{tm}}=\sum_{k,l=1}^{d} \left( \alpha_{k, l} f^{\dagger}_k f^{\phantom\dagger}_l + \beta_{k,l} \hat{n}_k \hat{n}_l \right),
\end{equation}
where $\alpha_{k,l}=\overline{\alpha}_{l,k}$ and the $\beta_{k,l}$ coefficients are real (moreover, at least one of the $\beta_{k,l}$ coefficients must be different from zero as the
Hamiltonian is assumed to be non-trivial and not quadratic). Let us first note that $\sum_{k,l}a_{kl}f^{\dagger}_k  f_l \in \Lie(\LOF)$. Moreover, we have 
\begin{equation}
 \sum_{k,l=1}^{d}  \beta_{k,l} \hat{n}_k \hat{n}_l = \sum_{k,l=1}^{d}  \beta_{k,l} ( \hat{n}_k-1/2) (  \hat{n}_l -1/2) +X' \ ,
\end{equation}
where $X'\in \Lie(\LOF)$. Therefore we obtain 
\begin{equation}\label{eq:alternativeTWOmode}
\langle \LOF, H_{\text{tm}} \rangle = \langle \LOF, H'_{\text{tm}} \rangle \ ,\text{ with }  H'_{\text{tm}}= \sum_{k,l=1}^{d}  \beta_{k,l} ( \hat{n}_k-1/2) (  \hat{n}_l -1/2)\ . 
\end{equation}
According to Theorem~\ref{thm:FERMHAM} (fermionic analogue of Theorem \ref{thm:BOSHAM} given in the main text) for $d\neq 2N$ we automatically get  $\langle \LOF, H'_{\text{tm}} \rangle=\U\left(\Hfer \right)$. From the same theorem it follows that for $d= 2N$ the universality is equivalent to the condition
\begin{equation}\label{eq:twoMODEex}
\left[H'_{\text{tm}} \otimes \I + \I \otimes H'_{\text{tm}}, \kb{\Psi_f}{\Psi_f}\right]\neq 0\ ,
\end{equation}
which is satisfied if and only if  $\ket{\Psi_f}$ (see Eq. \eqref{eq:LferVEC}) is not an eigenvector of $H'_{\text{tm}} \otimes \I + \I \otimes H'_{\text{tm}}$. For arbitrary Slater determinant $\ket{X}$ in $d/2$ particle sector, we have
\begin{equation}
H'_{\text{tm}}\ket{X}= \lambda_X \ket{X} \ .
\end{equation}  
Moreover, explicit computation using \eqref{eq:alternativeTWOmode} gives $\lambda_X = \lambda_{\bar{X}}$, for all $X$ of cardinality $d/2$. Consequently, we obtain
\begin{equation}
\left(H'_{\text{tm}} \otimes \I + \I \otimes H'_{\text{tm}}\right) \ket{X} \ot \ket{\bar{X}} = 2\lambda_X \ket{X} \ot \ket{\bar{X}} \ .
\end{equation}
Thus, by using Eq.~\eqref{eq:LferVEC}, we conclude that $\ket{\Psi_f} $ is an eigenvector of $H'_{\text{tm}} \otimes \I + \I \otimes H'_{\text{tm}}$ iff $\lambda_X= \lambda$,  for all $d/2$-element subsets of $[d]$. However, we assumed that $H_{\text{tm}}$ (and thus also $H'_{\text{tm}}$) is non-trivial, there must exist two Slater determinants $\ket{X_1}$ and $\ket{X_2}$ with $\lambda_{X_1} \ne \lambda_{X_2}$
(only the multiple of the identity in $\Hfer$ can have the same eigenvalue for all Slater determinants). Hence, Eq.~\eqref{eq:HtmX} cannot be satisfied, and consequently $\langle \LOF, H_{\text{tm}} \rangle=\U\left(\Hfer \right)$ also in the $d= 2N$ case.

\subsection{Example 4}\label{EX4:proof}
For Example 4, we proceed somewhat similarly to the proof of the previous example. To prove that the correlated hopping Hamiltonian $Y$ does not promote passive Fermionic Linear Optics to universality in the half-filling case, we have to show that for gates of the form $V=\exp(itY)$  (with arbitrary real $t$) the condition $[V \otimes V, \L_f]=0$ holds. As discussed in the previous example, this condition is equivalent to 
\begin{equation}\label{eq:fermEXcond}
[Y\ot\I +\I\ot Y , \L_f]=0\ .
\end{equation}
If the above holds, then Theorem~\ref{thm:FERMGATES} guarantees that $\langle \LOF, X\rangle=G_f$, which is a proper subgroup of $\U\left(\Hfer \right)$ (see also Theorem \ref{thm:FERMHAM} given in the next part of the Appendix).

We will start by showing that for $E_{1}=(\hat{n}_{2} + \hat{n}_4 - \hat{n}_{2} \hat{n}_4) \, f^{\dagger}_1  f^{\phantom\dagger}_{3}$ the equality $(E_{1} \otimes \I +\I \otimes E_{1}) \ket{\Psi_f} = 0$ is satisfied. Let us note that for any Slater determinant $\ket{X}=\ket{X(1),X(2),\ldots,X(d)}$ , we have that $E_{1} \ket{X} \ne 0$ iff $X(1) + X(3)=1$ and $X(2)+X(4)=1$. Now, using the form of $\ket{\Psi_f}$ given in Eq.~\eqref{eq:finalSLATER}, we obtain that 
\begin{align}
&\tbinom{2N}{N} \, (E_{1} \otimes \I + \I \otimes E_{1} ) \ket{\Psi_f}=(E_{1}\otimes \I + \I \otimes E_{1}) \sum_{\substack{X\subset[2N], \\  |X|=N}} \sgn(X) \ket{X} \ot \ket{\bar{X}}=
\\
&(E_{1} \otimes \I + \I \otimes E_{1} )\sum_{\substack{X'\subset \lbrace 5,\ldots,2N\rbrace, \\  |X'|=N-2}} 
 \widetilde{\sgn}( X')  \left( \ket{1100 \, X'}   \ket{0011 \, \bar{X}'} {+} \ket{0011 \, X'}  \ket{1100 \, \bar{X}'}{+}\ket{1001 \, X'}  \ket{0110 \, \bar{X}'} {+}\ket{0110 \, X'}  \ket{1001 \, \bar{X}'}\right),
\end{align}
where we introduced the notation $\widetilde{\sgn}(X')=\sgn(\{1,2\} \cup X')$ and used that from Eq. \eqref{eq:EXPLICITsgn} it trivailly follows   that $\sgn(\{1,2\} \cup X') = \sgn(\{3,4\} \cup X')=\sgn(\{1,4\} \cup X')=\sgn(\{2,3\} \cup X')$ for all $X'\subset \lbrace 5,\ldots, 2N \rbrace $.

Next, by straightforward calculations, one obtains
\begin{align}
& E_{1} \otimes \I  \left( \ket{1100 \, X'}   \ket{0011 \, \bar{X}'} {+} \ket{0011 \, X'}  \ket{1100 \, \bar{X}'}{+}\ket{1001 \, X'}  \ket{0110 \, \bar{X}'} {+}\ket{0110 \, X'}  \ket{1001 \, \bar{X}'}\right) \nonumber \\
&= \ket{1001 \, X'}   \ket{1100 \, \bar{X}'} - \ket{1100 \, X'}  \ket{1001 \, \bar{X}'},\\
& \I \otimes E_1 \left( \ket{1100 \, X'}   \ket{0011 \, \bar{X}'} {+} \ket{0011 \, X'}  \ket{1100 \, \bar{X}'}{+}\ket{1010 \, X'}  \ket{0101 \, \bar{X}'} {+}\ket{0101 \, X'}  \ket{1010 \, \bar{X}'}\right)= \nonumber \\
& = \ket{1100 \, X'}   \ket{1001 \, \bar{X}'} - \ket{1001 \, X'}  \ket{1100 \, \bar{X}'},
\end{align}
which immediately implies that $(E_1 \otimes \I +\I \otimes E_1) \ket{\Psi_f} = 0$. Analogous computations show that $(E^\dagger_1 \otimes \I +\I \otimes E^\dagger_1) \ket{\Psi_f} = 0$. Now, after defining
\begin{equation}
E_{2j-1}=(\hat{n}_{2j} + \hat{n}_{2(j+1)} - \hat{n}_{2j} \hat{n}_{2(j+1)}) \, f^{\dagger}_{2j-1}  f^{\phantom\dagger}_{2j+1}\  \text{and\ } F_{2j}=(\hat{n}_{2j-1} + \hat{n}_{2j+1} - \hat{n}_{2j-1} \hat{n}_{2j+1}) \, f^{\dagger}_{2j}  f^{\phantom\dagger}_{2(j+1)} \ ,
\end{equation}
and performing the completely analogous calculations one shows that \begin{equation}
(E^{\phantom\dagger}_{2j-1} \otimes \I +\I \otimes E^{\phantom\dagger}_{2j-1}) \ket{\Psi_f} = (E^{\dagger}_{2j-1} \otimes \I +\I \otimes E^{\dagger}_{2j-1}) \ket{\Psi_f}=  (F^{\phantom\dagger}_{2j} \otimes \I +\I \otimes F^{\phantom\dagger}_{2j}) \ket{\Psi_f} = (F^{\dagger}_{2j} \otimes \I +\I \otimes F^{\dagger}_{2j}) \ket{\Psi_f}=0\ .
\end{equation}Finally, using $Y=\sum_{j=1}^{d/2} (E^{\phantom\dagger}_{2j-1}+E^{\dagger}_{2j-1}+F^{\phantom\dagger}_{2j}+F^{\dagger}_{2j})$, we get that $(Y \otimes \I  + \I \otimes Y) \ket{\Psi_f}=0$, which implies \eqref{eq:fermEXcond}.

\subsection{Example 5}\label{EX5:proof}
We prove the assertion given in Example \ref{ex:FLO} by directly using  Theorem \ref{thm:FLOHAM}. From the point (a) of the aforementioned theorem it follows that for $d\neq 2 m$ the Hamiltonian $H_\mathrm{in}=m_1 m_2 m_3 m_4$ always promotes FLO to universality in $\Hfree$ (clearly, it does not belong to $\Lie(\FLO)$). On the other hand, from the point (c) of the same Theorem it suffices to show
\begin{equation}\label{eq:condEX}
\left[ m_1 m_2 m_3 m_4\ot \I + \I\ot m_1 m_2 m_3 m_4 , \L_\FLO \right]\neq 0  
\end{equation}  
in order to prove $\langle \FLO, H_{\mathrm{in}} \rangle = \U(\Hfree)$. In what follows we use the alternative representation $\L'_\FLO$ of the operator $\L_\FLO$ given in \eqref{eq:Lalt}. Explicit computation involving the commutation rules of Majorana-fermion operators give
\begin{gather}
\left[ m_1 m_2 m_3 m_4\ot \I + \I\ot m_1 m_2 m_3 m_4 , \L_\FLO \right]  =   \left[ m_1 m_2 m_3 m_4\ot \I + \I\ot m_1 m_2 m_3 m_4 , \L'_\FLO \right] \\
= \frac{1}{2^{2d-1}}  \left(\prod_{i=1, i\neq 4}^{2d-1} \left(\I\ot \I + m_{i} m_{i+1}\ot  m_{i} m_{i+1}\right) \right)  \left[ m_1 m_2 m_3 m_4\ot \I + \I\ot m_1 m_2 m_3 m_4 , \left(\I\ot \I + m_{4} m_{5}\ot  m_{4} m_{5}\right) \right]\\
= \frac{1}{2^{2d-1}}  \left(\prod_{i=1, i\neq 4}^{2d-1} \left(\I\ot \I + m_{i} m_{i+1}\ot  m_{i} m_{i+1}\right) \right) 2\left(m_1 m_2 m_3 m_5 \ot m_4 m_5 + m_4 m_5 \ot m_1 m_2 m_3 m_5\right)\ . \label{eq:finalEX4}
\end{gather}
The expansion of the product in \eqref{eq:finalEX4} shows that the whole expression  is manifestly nonzero and therefore the condition \eqref{eq:condEX} is satisfied.

\section{Proofs of technical theorems}\label{ap:mainTheorems}

In what follows we present the proofs of the results from the main text. We first prove theorems concerning Hamiltonian extension of various classes of linear optics. Later we prove results on the gate extension problem for these settings. 

\begin{lem}[Connected Lie group generated by a matrix Lie algebra \cite{hofmann2013}] \label{lem:LIEalgGRU}] For a Lie algebra $\k$
consisting  linear operators on $\H$,  $\k \subset \su(\H)$, we denote by $\langle \exp(\ii \k) \rangle$ a subgroup of $\SU(\H)$ generated by elements of the form $\exp(\ii X)$, for $X \in \k$. Then $\langle \exp(\ii \k) \rangle$ is a connected compact Lie subgroup of $\SU(\H)$  with the Lie algebra $\k$, i.e. $\Lie\left( \langle \exp(\ii \k) \rangle\right)=\k$.  

\end{lem}

\subsection{Hamiltonian extensions of linear optics}

\begin{customthm}{2}[Extensions of Passive Bosonic  Linear Optics via an additional Hamiltonian]\label{thm:BOSHAM2}
	Let $X\notin \Lie\left(\LOB\right)$ be a Hamiltonian acting on Hilbert space $\Hbos$ of $N$ bosons in $d$ modes. Let $\langle \LOB, X\rangle$ be the group of transformations generated by passive bosonic linear optics and $X$ in $\Hbos$. 
	We have the following possibilities:
	\begin{itemize}
		\item[(a)] If $d>2$, then $\langle \LOB, X\rangle=\U\left(\Hbos \right)$.
		\item[(b)] If $d=2$, $N\neq 6$, and $\left[X\ot \I + \I\ot X, \L_b\right]=0$, then 
		\begin{equation}
		\langle \LOB, X\rangle=G_b=\lbrace \left. U\in \U\left(\Hbos \right) \right|   \left[U\ot U, \L_b\right]=0\ \rbrace \ .
		\end{equation}
		\item[(c)] If $\left[X\ot \I + \I\ot X, \L_b\right]\neq0$, then  $\langle \LOB, X\rangle=\U\left(\Hbos\right).$
		
	\end{itemize}
		
\end{customthm}

\begin{proof}[Proof of Theorem \ref{thm:BOSHAM2}]
Let us first remark that since diagonal Hamiltonians are contained in $\Lie(\LOB)$ (see Eq.~\eqref{eq:LieBOS}), we can assume without loss of generality that $X$ is traceless, i.e., belongs to the Lie algebra $\su(\Hbos)$ of the special unitary group $\SU(\Hbos)$ (see Appendix \ref{ap:LieEx}). In what follows we describe the possible forms of the group $\GEN = \langle \Pi_b(\SU(d)),X\rangle$ generated by linear optical gates $\Pi_b(\SU(d))$ of determinant one and unitary evolutions of the form  $\exp(\ii t X)$. The results stated in the formulation of Theorem \ref{thm:BOSHAM2} will be then obtained by using the relation
\begin{equation}\label{eq:genRELBos}
\langle \LOB,X\rangle= \langle \GEN, \T(\Hbos) \rangle  ,
\end{equation}
where $\T(\Hbos)$ is the group of unitaries on $\Hbos$ of the form $\exp(\ii \theta) \I$, where $\theta\in\R$. The structure of $\GEN$  relies on the classification of maximal simple Lie subalgebras of classical simple compact  Lie algebras due to Dynkin \cite{Dynkin2000}. Recall that $\Lie\left(\Pi_b(\SU(d))\right)=\pi_b(\su(d))$ is the image of a representation of a compact simple  Lie algebra $\su(d)$. By definition (see Lemma~\ref{lem:LIEalgGRU}) the group $\GEN$ is a connected compact Lie group contained in $\SU(\Hbos)$, and we denote its Lie algebra by $\mathfrak{g}\subset\su\left(\Hbos\right)$.  Because $X\notin\pi_b(\su(d))$ and $X\in\su(\Hbos)$, we have the following  inclusions of compact connected Lie groups together with the inclusions of the corresponding Lie algebras
\begin{eqnarray}
\Pi_b(\SU(d))  &\subset \GEN & \subset \SU\left(\Hbos\right)\ , \\
\pi_b(\su(d))  &\subset \g &\subset \ \su\left(\Hbos\right)\ .
\end{eqnarray}
In \cite{Dynkin2000} it was proven that for $d\, {>}\,2$ the Lie algebra $\pi_b(\su(d))$ is a maximal Lie subalgebra in $\su\left(\Hbos\right)$. Therefore, by the assumption that $X\notin\pi_b(\su(d))$, we have $\g=\su\left(\Hbos\right)$.  Since $\GEN$ is the unique simply-connected Lie group with Lie algebra $\g$ we have $\GEN=\SU(\H_b)$.  This together with \eqref{eq:genRELBos}  concludes the proof of the case (a) since  $\langle \SU(\Hbos), \T(\Hbos) \rangle = \U(\Hbos) $.

We will now treat cases (b) and (c), focusing first on the situations with $d=2$ modes for $N\neq 6$. For such a situation we have the following maximal inclusions of Lie subalgebras \cite{Dynkin2000}
\begin{equation}\label{eq:inclBOSeven}
\pi_b(\su(2)) \subset \so\left(\Hbos \right) \subset \su\left(\Hbos\right)
\end{equation}
for even $N$, and
\begin{equation}\label{eq:inclBOSodd}
\pi_b(\su(2))  \subset \usp\left(\Hbos \right) \subset \su\left(\Hbos\right)
\end{equation}
for odd $N$. By "maximal inclusions" in the previous sentence we mean that
\begin{itemize}
\item[(i)] there are no other possible Lie algebras between $\pi_b(\su(2))$ and $\so\left(\Hbos\right)/ \usp\left(\Hbos \right)$;
\item[(ii)] there are no other possible Lie subalgebras between $\so\left(\Hbos\right)/ \usp\left(\Hbos \right)$ and $\su(\Hbos)$.
\end{itemize} 
Moreover, from Table 5 in \cite{Dynkin2000} it follows that for $N\neq 6$   there is no other Lie algebras between $\pi_b (\su(2))$ and $\su(\Hbos)$.  The Lie groups $\SO\left(\Hbos \right)$ and $\USP\left(\Hbos \right)$ are defined as subgroups of $\SU\left(\Hbos\right)$ that preserve, respectively, the symmetric and the antisymmetric non-degenerate real bilinear form $B$ on $\Hbos$ (see Lemma \ref{Projbosons})
\begin{equation}
B\left( \ket{\psi},\ket{\phi}\right )\defeq \bra{\Psi_b}(\ket{\psi}\ot\ket{\phi})\ .
\end{equation}
From the above discussion we can deduce, repeating the reasoning given in the proof of (a), that there are three possibilities:
\begin{itemize}
	\item [(i)] $\GEN=\SO(\Hbos)$ (for even $N\neq 6$);
	\item [(ii)] $\GEN=\USP(\Hbos)$ (for odd $N$); 
	\item [(iii)]  $\GEN=\SU(\Hbos)$  (for $N\neq 6$).
\end{itemize}
The cases (i) and (ii) occur if and only if unitaries $V_t=\exp(\ii t X)$ preserve the form $B$ (for suitable $N$).  The preservation of $B$ is equivalent to $V_t\ot V_t \ket{\Psi_b} = \ket{\Psi_b}$, which on the Lie algebra level translates to 
\begin{equation}\label{eq:Preserv}
\left( X\ot\I +\I\ot X\right) \ket{\Psi_b} = 0 .  
\end{equation} 
Under the assumption $\tr(X)=0$ the above condition equivalent to 
\begin{equation}\label{eq:BOScomPR}
\left[X\ot\I +\I\ot X , \kb{\Psi_b}{\Psi_b} \right]=0\ .
\end{equation}
We have therefore proved point (b) of the Theorem. If the commutator \eqref{eq:BOScomPR} does not vanish we know that unitaries $V$ do not preserve $B$ and therefore we are in the case (iii). Using \eqref{eq:genRELBos} we can therefore conclude 
the proof of the point (c) of the Theorem for $N\neq 6$.

The situation for $N\, {=} \, 6$ bosons in $d\, {=} \,2$ modes is more complicated as in this case $\pi_b(\su(2))$ is not maximally embedded in $\so(\Hbos)$, instead it is maximal in the seven dimensional representation of the Lie algebra of the exotic simple Lie algebra $\g_2$, which itself is a maximal Lie subalgebra of $\so(\Hbos)$,
\begin{equation}\label{eq:simpleC7list}
\pi_b(\su(2)) \subset \g_2 \subset \so(\Hbos) \subset \su(\Hbos)\ .
\end{equation} 
The corresponding inclusions on the Lie group level have the form
\begin{equation}
\Pi_b(\SU(2)) \subset \mathrm{G}_2 \subset \SO(\Hbos) \subset \SU(\Hbos)\ .
\end{equation} 
First, we observe that $\GEN$ is a connected compact Lie group (this follows form Lemma \ref{lem:LIEalgGRU}). Moreover, its Lie algebra $\g$ must be simple (this is because $\g$ contains $\pi_b(\su(2))$, whose elements commute only with the identity on $\Hbos$, due to the irreducibility of $\pi_b$).  On the other hand form the results of Dynkin (see Table 5 of \cite{Dynkin2000} ) it follows that all simple Lie  subalgebras of $\su(\Hbos)$ (for $d=2$ and $N=6$ dimension of $\Hbos$ is 7) are the ones appearing in the sequence of inclusions \eqref{eq:simpleC7list}. Since $\mathrm{G}_2 \subset \SO(\Hbos)$ we can use the reasoning analogous to the one given in the proof of (c) for $N\neq6$ to conclude that $\left[X\ot\I +\I\ot X , \kb{\Psi_b}{\Psi_b} \right]\neq$ implies $\GEN=\SU(\Hbos)$. Using this and \eqref{eq:genRELBos}  we finish the proof.
\end{proof}
\begin{customthm}{5}[Extensions of Passive  Fermionic Linear Optics via additional Hamiltonian]\label{thm:FERMHAM}
	Let $X\notin \Lie\left(\LOF\right)$ be a Hamiltonian acting on Hilbert space of $N$ fermions in $d$ modes $\Hfer$,  where $N\notin\lbrace 0,1,d-1,d\rbrace$. Let $\langle \LOF, X\rangle$ be the group  generated by passive fermionic linear optics and $X$ in $\Hfer$. For $d=2N$ (half-filling) let $L_f\in\Herm\left(\Hfer\ot\Hfer\right)$ be defined as in Eq.~\eqref{eq:Lfer}. 	We have the following possibilities:
	\begin{itemize}
		\item[(a)] If $d\neq 2N$, then $\langle \LOF, X\rangle=\U\left(\Hfer\right)$ .
		\item[(b)] If $d=2N$ and $\left[X\ot\I+\I\ot X, \L_f\right]=0$, then
	\begin{equation}
		\langle \LOF, X\rangle=G_f=\lbrace \left. U\in \U\left(\Hfer \right) \right|    \ [U\ot U, \L_f]=0\ \rbrace \ .
		\end{equation}
		\item[(c)] If $d=2N$ and $[V\ot V, \L_f]\neq 0$, then  $\langle \LOF, X\rangle =\U\left(\Hfer\right)$.
	\end{itemize}
\end{customthm}

\begin{customthm}{6}[Extensions of Active Fermionic Linear Optics via additional Hamiltonian]\label{thm:FLOHAM}
	Let $X\notin \Lie\left(\FLO\right)$ be a  Hamiltonian acting on positive-parity subspace of fermionic $d$-mode Fock space $\Hfree$ with $d>3$. Let $\langle \FLO, X\rangle$ be the group of transformations generated by active fermionic linear optics and $X$ in $\Hfree$. For $d=2m$ let $\L_{\FLO}\in\Herm\left(\Hfer\ot\Hfer\right)$ be defined as in Eq.\eqref{eq:LFLO}.
	We have the following possibilities:
	\begin{itemize}
		\item[(a)] If $d\neq 2k$, then $\langle \FLO, X\rangle=\U\left(\Hfree\right)$ .
		\item[(b)] If $d=2k$ and $[X\ot\I +\I\ot X , \L_{\FLO}]=0$, then 
		\begin{equation}
	\langle \FLO, X\rangle= 	G_{\FLO}=\lbrace \left. U\in \U\left(\Hfree \right) \right|    [U\ot U, \L_{\FLO}]=0\ \rbrace \ .
		\end{equation}
		\item[(d)] If $d=2k$ and $[X\ot\I +\I\ot X , \L_{\FLO}]\neq 0$, then  $\langle \FLO, X\rangle=\U\left(\Hfree\right)$.	
	\end{itemize}
\end{customthm}
\begin{proof}[Proofs of Theorems \ref{thm:FERMHAM} and \ref{thm:FLOHAM}]
The proofs of both theorems are very similar and analogous to the proof for the bosonic passive linear optics. For this reason we  present them together. In what follows we will use $K$ to denote passive ($\LOF$) or active ($\FLO$) fermionic linear optical transformations. By $K_s$ we will denote the "simple" part of $K$ i.e. a subgroup of $K$ consisting  of operators having unit determinant on the appropriate fermionic Hilbert space $\H$ (equal to $\Hfer$ or $\Hfree$ respectively). The symbols $\k$ and $\k_s$ will denote the Lie algebras of $K$ and $K_s$ respectively.  The main steps are analogous to the ones form the proof  Theorem \ref{thm:BOSHAM2}.

\begin{enumerate}
\item Without the loss of generality we can assume that $\tr(X)=0$. After setting $\GEN= \langle K_s, X \rangle$ we have
\begin{equation}
\langle K, X \rangle = \langle \GEN,\T(\H) \rangle\ .
\end{equation} 

\item For the cases specified in points (a) of Theorems \ref{thm:FERMHAM} and \ref{thm:FLOHAM} we have that $\k_s \subset \su(\H)$ are maximal simple subalgebras of $\su(\H)$ (see Table 5 of \cite{Dynkin2000}). Therefore, by repeating the reasoning analogous to the bosonic case, we get that under the assumptions given in point (a) we have $\langle K, X \rangle = \U(\H)$.

\item Analogously to the case of $\LOB$, for situations specified in the assumptions of points (b) and (c), the appropriate  representations $\Pi_f$ and $\Pi_{\FLO}^+$ preserve bilinear forms given by $B\left( \ket{\psi},\ket{\phi}\right )= \bra{\Psi} \ket{\psi}\ot  \ket{\phi}$, where, depending on the context, $\ket{\Psi}=\ket{\Psi_f}$ or  $\ket{\Psi}=\ket{\Psi_\FLO}$ (vectors  $\ket{\Psi_f},\ket{\Psi_\FLO}$ have been given explicitly in Lemma \ref{lem:ProjPassiveDiff} and Lemma \ref{lem:ExplictitProjFLO}).
 
\item Depending on the value of $d$,  bilinear forms $B$ are either symmetric or antisymmetric, giving rise to special orthogonal  or unitary symplectic  groups  ($\SO(\H)$ or   $\USP(\H)$) on $\H$. However, in the considered situations it is always the case (see again Table 5 from \cite{Dynkin2000}) that $\k_0$ is the maximal subalgebra in $\so(\H)$ or $\usp(\H)$ and there are no other Lie subalegbras of $\su(\H)$ containing $\k_s$.  This allows us to repeat essentially  the whole reasoning from the bosonic scenario. Once once again the commutator $\left[X\ot\I +\I\ot X , \kb{\Psi_b}{\Psi_b} \right]$ plays the crucial role in deciding of the structure of $\GEN$ and therefore $\langle K, X\rangle$.  The resulting  conditions are given precisely in points (b) and (c) of Theorems \ref{thm:FERMHAM} and \ref{thm:FLOHAM}.

\end{enumerate}
\end{proof}

\subsection{Gate extensions of linear optical gates}

It turns out that most o the results concerning  the Hamiltonian extensions of linear optical gates carry over to the gate- extension problems. Before we prove the main theorems we first give a number of auxiliary lemmas concerning the structure of the gates normalising linear optical transformations or related gate-sets. It turns out that the normalising gates play a crucial role in the proofs of results concerning gate extension problems.

\begin{lem}[Normalising gates for orthogonal and unitary symplectic gates]\label{lem:normBILIN}
Let $K=\SO(\H)$ or $K=\USP(\H)$ (for even dimensional $\H$). 
Let $V\in\U(\H)$ be the gate satisfying $V K V^\dagger = K$, that is for all $U\in K$  we have $V U V^\dagger = U'$, where $U'\in K$. Then, we have $V=\exp(\ii \theta) V'$, for $\theta \in\R$ and $V'\in K$.    
\end{lem}

\begin{proof}
The preservation of the invariant orthogonal or symplectic structure corresponds to preservation of the fixed vector $\ket{\Psi}\in \H \ot \H$ (whose symmetry decides whether we are dealing with orthogonal or symplectic group). In other words $\U\in K$ if and only if $U\ot U \ket{\Psi}=\ket{\Psi}$. Moreover, the vector $\ket{\Psi}$ is defined uniquely up to a global phase as the inly vector from $\H\ot \H$ satisfying $U\ot U \ket{\Psi}=\ket{\Psi}$ for all $U\in K$. Now, consider a gate $V$ that normalises $K$ and an arbitrary $U\in K$. We have 
\begin{equation}
(U\ot U)(V \ot V)  \ket{\Psi}= V\ot V (V^\dagger \ot V^\dagger)  (U\ot U)(V \ot V)  \ket{\Psi}=  V \ot V  \ket{\Psi}\ ,
\end{equation}
where we have used the fact that $V U V^\dagger \in K$. Since $U$ was arbitrary we conclude that $V\ot V \ket{\Psi} = \exp(\ii \theta)  \ket{\Psi}$ and therefore $V=\exp(\ii \theta') U'$, for suitable $\theta'\in\R$ and $U'\in K$. 
\end{proof}

\begin{rem}\label{rem:renNORM}
From the character theory it follows that compact groups irreducibly represented on $\H$ have at most one trivial representation in $\H \ot \H$. Therefore, if a given irreducible representation $\Pi$ preserves orthogonal or symplectic structure (defining the corresponding group $K$), the corresponding invariant form is also uniquely defined. By a slight modification of the above argument we get that if an irreducible representation $\Pi$ of a compact group $G$ is real or quaternionic and if for a unitary $V$ we have $V \Pi(G) V^\dagger = \Pi(G)$, then $V=\exp(\ii \theta) W$ , where $W\in K$ and $\theta\in\R$.
\end{rem}

\begin{lem}[Normalising gates for passive Bosonic linear optics]\label{lem:normBOS}
Let $V\in\U(\Hbos)$ be the gate satisfying $V \Pi_b(\SU(d)) V^\dagger =  \Pi_b(\SU(d))$, that is for all $U\in\Pi_b (\SU(d))$  we have $V U V^\dagger = U'$, where $U'\in\Pi_b(\SU(d)$. Then, we have $V=\exp(\ii \theta) V'$, for $\theta \in\R$ and $V'\in\Pi_b (\SU(d))$.    
\end{lem}  

\begin{lem}[Normalising gates form passive Fermionic linear optics: generic case]\label{lem:normFERgen}
Let $V\in\U(\Hfer)$ be the gate satisfying $V \Pi_f(\SU(d)) V^\dagger = \Pi_f (\SU(d))$, that is for all $U\in\Pi_f (\SU(d))$  we have $V U V^\dagger = U'$, where $U'\in\Pi_f(\SU(d)$. Moreover, assume that $2N\neq d$.  Then, we have $V=\exp(\ii \theta) V'$, for $\theta \in\R$ and $V'\in\Pi_f (\SU(d))$. 
\end{lem} 

\begin{lem}[Normalising gates form passive Fermionic linear optics: half-filling case]\label{lem:normFERex}

Let $V\in\U(\Hfer)$ be the gate satisfying $V \Pi_f(\SU(d)) V^\dagger = \Pi_f (\SU(d))$, that is for all $U\in\Pi_f (\SU(d))$  we have $V U V^\dagger = U'$, where $U'\in\Pi_f(\SU(d)$. Moreover, assume that $2N= d$.  Then, we have $V=\exp(\ii \theta) W V'$ \emph{or} $V=\exp(\ii \theta) V'$  for $\theta \in\R$, $V'\in\Pi_f (\SU(d))$  and $W=\prod_{i=1}^d (f_i + f^\dagger_i)$.
\end{lem}

\begin{lem}[Normalising gates form active Fermionic linear optics]\label{lem:normFLO}
Let $V\in\U(\Hfree)$ be the gate satisfying $V \Pi^+_\FLO(\Spin(2d)) V^\dagger = \Pi^+_\FLO(\Spin(2d))$, that is for all $U\in\Pi^+_\FLO(\Spin(2d))$  we have $V U V^\dagger = U'$, where $U'\in\Pi^+_\FLO(\Spin(2d))$. Moreover, assume that $d\geq 4$.  Then, we have $V=\exp(\ii \theta) V'$, for $\theta \in\R$ and $V'\in\Pi^+_\FLO(\Spin(2d))$. 
\end{lem}

Proofs of the above lemmas are to large extent similar and we present them together.

\begin{proof}[Proof of Lemmas \ref{lem:normBOS}, \ref{lem:normFERgen}, \ref{lem:normFERex} and \ref{lem:normFLO}]
Let  $V\in\U(\H)$ normalize an irreducible representation of a compact Lie group $\Pi(K)\subset \U(\H)$. Since any  one-parameter subgroup of this group is mapped to another one-parameter subgroup of $\Pi(K)$, one obtains
\begin{equation}
V \pi (\k) V^{\dagger}\ =\pi (\k)\ ,
\end{equation}
where $\k=\Lie(K)$ and $\pi$ is the induced representation of of the $\k$ in $\H$. Let us moreover assume that the representation $\pi$ is the faithful representation of $\k$ (this is the case for the considered groups representations). We can thus infer that for each $X \in \k $ there exists a unique $X'\in \k$ such that
\begin{equation}
V\pi(X) V^{\dagger}=\pi(X')\ .
\end{equation}
Therefore, one can define a unique mapping $\tau_V : \k \to \k$ with the property that for all $X\in\k$ $V\pi(X) V^{\dagger}=\pi(\tau_V(X))$. Due to the invertibility of $V$, the map $\tau$ is bijective.  It is also by definition linear and, from the identity $V\pi([X_1, X_2]) V^{\dagger}= [V\pi(X_1)V^{\dagger}, V\,\pi(X_2)V^{\dagger}]$, one can infer that $[\tau_V(X_1), \tau_V(X_2)]=\tau_V([X_1, X_2])$. Thus, $\tau_V$ is an \emph{automorphism} of $\k$. We will use this information to prove statements given in Lemmas \ref{lem:normBOS}, \ref{lem:normFERgen}, \ref{lem:normFERex} and \ref{lem:normFLO} (in almost all cases appearing in these lemmas \footnote{The exception is the case $\k=so(4)$ appearing for $\FLO$ for $d=2$ modes. For this case we have $\pi^+_\FLO(\so(4))\approx \su(\Hfree)$.} transformations for which  in this lemmas we have irreducible faithful representation of simple Lie algebras). The fallowing arguments are based on the structure of the automorphism groups of classical simple Lie algebras which can be found in Section 16.5 of \cite{Humphreys1972}.

\begin{itemize}

\item[(I)] Bosonic representation of $\SU(d)$ - Lemma \ref{lem:normBOS}. We first note that for the case of of $\su(2)$ the group  of Lie algebra automorphisms $\Aut(\su(2))$ consists only of inner automorphisms i.e. 
\begin{equation}
\Aut(\su(2))=\Inn(\su(2))=\SET{\Ad_{U}}{\Ad_U(X)=U X U^\dagger,\  U\in\SU(2)}\ .
\end{equation}
Therefore, we have $\tau_V = \Ad_U$ for some $U\in\SU(2)$. Consequently, we have that for all $X\in \su(2)$ the following chan of inequalities hold
\begin{equation}\label{eq:chainEQ}
V \pi_b(X) V^\dagger = \pi_b(\tau_V (X))= \pi_b ( U X U^\dagger ) = 
\Pi_b ( U) \pi_b(X) \Pi_b(U)^\dagger\ ,
\end{equation}
where in the last equality we used the standard relation between the representation of Lie group $\Pi$ and the induced representation of its Lie lagebra $\pi$. Comparing expressions on the both sides of \eqref{eq:chainEQ} we obtain
\begin{equation}
\Pi_b(U)^\dagger V  \pi_b(X) = \pi_b(X) \Pi_b(U)^\dagger V\ , \text{for all } X\in\su(2)\ .  
\end{equation}
Consequently, by the virtue of Shur's lemma \cite{FultonHarris} (a representation $\pi_b$ is irreducible) , we have $\Pi_b(U)^\dagger V = \exp(\ii \theta) \I$, for $\theta\in\R$. This completes the proof in this case. 

If the number of modes is greater then two ($d>2$), the automorphisms group consists of two disjoint parts
\begin{equation}\label{eq:AutSUD}
\Aut(\su(d))=\Inn(\su(d)) \cup \Inn(\su(d)) \cdot \lbrace \alpha \rbrace
\end{equation}
where $\Inn(\su(d))=\SET{\Ad_{U}}{\Ad_U(X)=U X U^\dagger,\  U\in\SU(d)}$ and $\alpha(X) = - X $ is a non -inner automorphism. It is now important to know that for every representation $\pi$ we have $\pi\circ \alpha \approx \bar{\pi}$, where $\bar{\pi}$ denotes the dual representation of $\su(d)$ and $"\approx"$ denotes the equivalence of representations. Now, we need to consider two cases (i) $\tau_V \in  \Inn(\su(d))$ and (ii) $\tau_V \in \Inn(\su(d)) \cdot \lbrace \alpha \rbrace$. For the case (i) we can essentially repeat the reasoning given for the case of $\su(2)$ and conclude that $V = \exp(\ii \theta) \Pi_b(U)$, for some $\theta\in\R$ and $U\in\SU(d)$. On the other hand, the situation (ii) is impossible since this would imply that $\pi_b \approx \bar{\pi}_b$, which happens only when $d=2$.

\item[(II)]  Fermionic representation of $\SU(d)$: generic case - Lemma \ref{lem:normFERgen}. For this case the proof is easy since we can adopt again the strategy given in point (I). Concretely, we must have that $\tau_V = \Ad_U$, for some $U\in\SU(d)$. If it was not the case then we would have $\pi_f \approx  \bar{\pi}_f$, which happens for the case of half filling. Knowing that $\tau_V = \Ad_U$, for some $U \in \SU(d)$, we conclude that $V=\exp(\ii \theta) \Pi_f(U')$, for $\theta\in\R$ and $U'\in\SU(d)$.

\item[(III)]  Fermionic representation of $\SU(d)$: half-filling case -  Lemma \ref{lem:normFERex}. For this situation we actually have $\pi_f \approx \bar{\pi}_f$ and therefore we cannot exclude the possibility that $\tau_V = \Ad_U \alpha$, for $U\in\SU(d)$. Let us assume first that we have another $V' \in \U(\Hfer)$ normalising $\Pi_f (\SU(d))$  and satisfying $\tau_{V'} = \Ad_{U'} \alpha$. We then have $V' = \exp(\ii \theta) \Pi_f(U'')  V$ and thus the non-trivial normalising gate $V$ is essentially unique \footnote{In order to prove this assertion we notice that from $\tau_{V'} = \Ad_{U'} \alpha$ and  $\Ad_U \alpha$ it follows that $ \Ad_{U'^\dagger} \tau_{V'}= \Ad_{U^\dagger} \tau_{V}$. We can now use the reasoning analogous to the one given in point (I) to prove that $V' = \exp(\ii \theta) \Pi_f(U U'^\dagger) V$.} (up to a global phase and the action of $\Pi_f (U'')$, for $U''\in\SU(d)$).  Therefore we have two possibilities: (i) $\tau_V = \Ad_U$ or (ii) $\tau_V = \Ad_U \cdot \alpha$. The first possibility again results in $V=\exp(\ii\theta)\Pi_f(U)$. On the other the explicit computation using the fact that we are in the half-filling scenario $2N=d$)  shows that the particle-whole duality $W= \prod_{i=1}^d (f_i +f_{i}^\dagger)$ effectively realises the non-trivial automorphism i.e. 
\begin{equation}
W \pi_f (X) W^\dagger= - \pi_f (X)\ , \text{for all } X\in\su(d)\ .
\end{equation}
This observation, together with the uniqueness of the normalising gate (up to the action of elements form $\Pi_f(\SU(d))$ and a global phase factor) finishes the proof.

\item[(IV)] Positive-parity spinor representation of $\Spin(2d)$ - Lemma \ref{lem:normFLO}. 
If the number of modes is greater then four ($d>4$), the automorphisms group of $\so(2d)$ consists of two disjoint parts
\begin{equation}\label{eq:AutSUD}
\Aut(\so(2d))=\Inn(\so(2d)) \cup \Inn(\so(2d)) \cdot \lbrace \beta \rbrace\ ,
\end{equation}
where $\Inn(\so(2d))=\SET{\Ad_U}{\Ad_U (X)=U X U^\dagger\ , U\in \Spin(2d)  }$. Thus we again have two possibilities: either (i) $\tau_V =\Ad_U$, for $U\in\Spin(2d)$ or (ii)  $\tau_V = \Ad_U \cdot \beta $ for a non-inner automorphism $\beta$ of a Lie algebra $\so(2d)$. The case (i) implies again (see reasoning given in point (I) above) that $V=\exp(\ii \theta) \Pi_{\FLO}^+(U)$, for $\U\in\Spin(2d)$ and $\theta\in\R$.  Let us now settle the possibility (ii). For $d>4$ the outher automorphism $\beta$ of $\so(2d)$ can be taken to be of the form 
\begin{equation}\label{eq:outherSO}
\beta(E_{ij}) = \begin{cases}
    E_{ij}\ ,& \text{if } j\neq 2d\\
   -  E_{ij}\ ,              & \text{if } j=2d \end{cases}\ ,
\end{equation}
where $E_{ij}$ ($i<j$)  are the basis elements of $\so(2d)$ given in \eqref{eq:BASISso2D}. For the positive parity representation $\pi_{\FLO}^+$ we have the identity
\begin{equation}
\prod_{i=1}^d \pi^+_\FLO(E_{2i-1,2i}) = \frac{\I}{2^d}\ .
\end{equation}
Conjugating the above equation by $V$ and using the definition of $\tau_V$ we obtain
\begin{equation}\label{eq:IdentityFLO}
\prod_{i=1}^d \pi(\tau_V(E_{2i-1,2i})) = \frac{\I}{2^d}\ .
\end{equation}
Assuming that $\tau_V = \Ad_U \cdot \beta$ and using \eqref{eq:outherSO} we obtain 
\begin{equation}
\prod_{i=1}^d \pi^+_\FLO(\tau_V(E_{2i-1,2i})) = \left(\prod_{i=1}^{d-1} \pi^+_\FLO(E_{2i-1,2i}) \right)(-\pi^+_\FLO(E_{2i-1,2i})) = - \frac{\I}{2^d}\ ,
\end{equation}
which contradicts \eqref{eq:IdentityFLO}. We then necessarily must have  $\tau_V = \Ad_U$ for some $U\in\Spin(2d)$.

The result for the positive-parity spinor representantion of $\Spin(8)$ for $d=4$ modes follows easily from the Lemma \ref{lem:normBILIN} as for this case we have $\Pi^+_\FLO (\Spin(8)) \approx \SO(8)$ \cite{Oszmaniec2012,Oszmaniec2014}.

\end{itemize}

\end{proof}

In the following proofs we will need  the three well-known technical statements.
\begin{lem}[Extension of a Lie subgroup \cite{hofmann2013}]
Let $K\subset\U(\H)$ be a compact Lie group. Le $V\in \U(\H)$ be a unitary  operator that does not belong to $K$ i.e. $V\notin K$. Then, the group generated by $K$ and $V$, $\langle K, V\rangle$, is the smallest compact Lie subgroup of $\U(\H)$ containing both $K$ and $V$. 
\end{lem}

\begin{lem}[Structure of compact Lie groups \cite{hofmann2013}]\label{lem:structureCOMPACTlie}
Let $K\subset\U(\H)$ be a compact Lie group. Then $K$ has the following structure $K = K_0 \rtimes  D$, where $K_0$ is the maximal connected component of $K$ containing $\I$ and $D$ is a finite discrete group. Moreover, $K_0$ is a normal subgroup of $K$.
\end{lem}

\begin{lem}[Extension by a non-normalising element]\label{lem:normEXT}
Let $K\subset\U(\H)$ be a connected compact Lie subgroup of $\U(H)$. Let $V\in\U(\H)$  be a unitary gate that does not normalise $K$ and let $\GEN =\langle K, V\rangle$ be the Lie group generated by the $K$ and $V$ (see Lemma \ref{lem:structureCOMPACTlie}).   Then, we have a strict inclusion
\begin{equation}\label{eq:strictINCL}
\Lie(K) \subset \Lie(\GEN_0)
\end{equation}
where $\GEN_0$ is a connected component of $\GEN$. In other words we have $\dim(\Lie(K)) < \dim(\GEN_0)$.
\end{lem}
\begin{proof}
We have $V K V^\dagger \subset \GEN$ but $K \neq V K V^\dagger$. Moreover, $V K V^\dagger$ is a subset of $\GEN_0$, the connected component of $\GEN$.  Since the mapping $U\rightarrow V U V^\dagger$ is a diffeomorphism in $\U(\H)$, $\Lie(K)$ is mapped bijectively to $\Lie(V KV^\dagger)$. However, we also have $\Lie(K) \neq \Lie(V K V^\dagger)$, since, because $K$ is connected,  these Lie algebras uniquely specify the groups $K$ and $V K V^\dagger$. Consequently,  
\begin{equation}\label{eq:extensionOFang}
\Lie(K) \subset \Lie(K) \cup \Lie(V K V^\dagger) \subset \Lie(\GEN_0)\ ,
\end{equation}
and therefore Lie algebra of $\GEN_0$ must be strictly bigger then $\Lie(K)$.
\end{proof}

\begin{customthm}{1}[Extensions of Passive Bosonic  Linear Optics with an additional gate]\label{thm:BOSGATES2}
Let $V\notin \LOB$ be a gate acting on the Hilbert space $\Hbos$ of $N>1$ bosons in $d$ modes; and let $\langle \LOB, V\rangle$ be the group of transformations generated by passive bosonic linear optics and $V$ in $\Hbos$. We have the following possibilities:
\begin{itemize}
\item [(a)] If $d>2$, then $\langle \LOB,V\rangle =\U\left(\Hbos \right)$.
\item [(b)] If $d=2$, $N\neq 6$ and $[V\ot V, \L_b]=0$, then  $
\langle \LOB,V\rangle=G_b= \lbrace \left. U\in \U\left(\Hbos \right) \right|    \left[U\ot U, \L_b\right]=0\ \rbrace$ .

\item [(c)] If $d=2$ and $[V\ot V, \L_b ]\neq 0$, then  $\langle \LOB,V\rangle=\U\left(\Hbos\right)$.
\end{itemize}
\end{customthm}

\begin{proof}
Similarly to the case of the Hamiltonian extension we note that $\LOB=\langle \Pi_b(\SU(d)),\T(\Hbos)\rangle$ and we can assume without the loss of generality that $V\in\SU(\Hbos)$.  We then have 
\begin{equation}\label{eq:convenientGEN}
\langle \LOB, V \rangle = \langle \GEN , \T(\Hbos) \rangle\ ,\ \text{where }\ \GEN =\langle \Pi_b(\SU(d)) , V \rangle\ .
\end{equation}
For this reason in what follows we will focus  on the properties of $\GEN$.  From Lemma \ref{lem:normBOS} we see that $V$ cannot normalise $\Pi_b(\SU(d))$. By by the virtue of Lemma \ref{lem:normEXT} we have a strict inclusion $\pi_b (\su(d)) \subset \Lie(\GEN_0)$. We now give the specific proofs of different cases given in the theorem.
\begin{itemize}
\item[(a)] From the proof of Theorem \ref{thm:BOSHAM2} we have that for $d>2$ a Lie algebra $\pi_b (\su(d))$ is a maximal Lie sub algebra of $\su(\H)$. Therefore $\Lie(\GEN_0)=\su(\Hbos)$ and consequently $\GEN_0 = \GEN = \SU(\Hbos)$. Form this an \eqref{eq:convenientGEN} the  point (a) follows.
\item[(b)] From the proof of point (b) of Theorem \ref{thm:BOSHAM2} we know that for $d=2$ and $N\neq 6$ the Lie algebra $\pi_b (\su(2))$ is a maximal Lie subalgebra in $\so(\Hbos)$ or $\usp(\H)$ (depending whether $N$ is odd or even). The condition $[V\ot V, \L_b$ is then equivalent to $\GEN \subset \SO(\Hbos)$ (for even $N$) or $\GEN \subset \USP(\Hbos)$ (for odd $N$). Because of the maximality of the inclusions of $\pi_b (\su(2))$ in the Lie algebras of these groups we conclude $\GEN =\SO(\Hbos)$ (for even $N$) and $\GEN = \USP(\Hbos)$ (for odd $N$). Using \eqref{eq:convenientGEN} we conclude the proof of (b).
\item [(c)] From Lemma \ref{lem:structureCOMPACTlie} it follows that $\GEN=\GEN_0 \rtimes D$, where $\GEN_0$ is simply-connected Le group and $D$ is a discrete group that normalizes $\GEN_0$.   In the proof of Theorem \ref{thm:BOSHAM2} we saw that the only Lie algebras in $\su(\Hbos)$ containing $\pi_b(\su(2))$ are (i) $\so(\Hbos)$ (for even $N$), (ii) $\usp(\Hbos)$ (for odd $N$)), (iii) $\g_2$ (contained in $\so(\Hbos)$, appearing only for $N=6$, and (iv) $\su(\Hbos)$ itself. Condition $[V\ot V,\L_b]   \neq 0$ implies that $V$ is not contained in any of the connected subgroups of $\SU(\Hbos)$: $\SO(\Hbos)$, $\USP(\Hbos)$, or $\mathrm{G}_2$. Moreover, by the virtue of Lemma \ref{lem:normBILIN} and Remark \ref{rem:renNORM} $V$ cannot  normalise any of these groups unless $V=\exp(\ii \theta) \I$ or $V=\exp(\ii \theta) W$, with $W\in\SO(\Hbos)$ (for $N=6)$). However, these two possibilities contradict the assumption $[V\ot V,\L_b]\neq 0$.  We therefore conclude that $\GEN_0=\SU(\H)$, which together with \eqref{eq:convenientGEN} finishes the proof.
\end{itemize}

\end{proof}

\begin{customthm}{3}[Extensions of Passive  Fermionic Linear Optics with an additional gate]\label{thm:FERMGATES2}
Let $V\notin \LOF$ be a gate acting on Hilbert space of $N$ fermions in $d$ modes $\Hfer$, where $N\notin\lbrace 0,1,d-1,d\rbrace$. Let $\langle \LOF, V\rangle$ be the group of transformations generated by passive fermionic linear optics and $V$ in $\Hfer$.  We -have the following possibilities:
\begin{itemize}
\item[(a)] If $d\neq 2N$, then $\langle \LOF, V\rangle=\U\left(\Hfer\right)$.
\item[(b)] If $d=2N$ and  $V= W k$, for $k\in\LOF$ and  $W=\prod_{i=1}^d (f_i + f^\dagger_i)$, then $\langle \LOF, V\rangle = \LOF \cup \LOF \cdot W$.
\item[(c)] If $d=2N$, $V\neq gW $, for $g\in\LOF$, and $[V\ot V, \L_f]=0$, then $ \langle \LOF, V\rangle = G_f=\lbrace \left. U\in \U\left(\Hfer \right) \right|    \ [V\ot V, \L_f]=0\ \rbrace$.
\item[(d)] If $d=2N$ and $[V\ot V, \L_f]\neq 0$, then  $\langle \LOF, V\rangle=\U\left(\Hfer\right)$.
\end{itemize}
\end{customthm}
\begin{proof}
Except for the case (b) the strategy of the proof for $\LOF$ is analogous to the proof of the bosonic case. Just like in that scenario we have $\LOF=\langle \Pi_f(\SU(d)),\T(\Hfer)\rangle$ and we can assume without the loss of generality that $V\in\SU(\Hfer)$.  We thus obtain, 
\begin{equation}\label{eq:convenientGENfer}
\langle \LOF, V \rangle = \langle \GEN , \T(\Hfer) \rangle\ ,\ \text{where }\ \GEN =\langle \Pi_f(\SU(d)) , V \rangle\ .
\end{equation}
By the virtue of Lemma \ref{lem:normFERgen} for the cases (a), (c) and (d) the normaliser of $\Pi_f(\SU(d))$ is either trivial of $V$ does not normalise $\Pi_f(\SU(d))$. Furthermore, one can prove points (a), (c) and (d) by essentially mimicking the resining given in the proofs points (a), (b) and (c)  of Theorem \ref{thm:BOSGATES2}.

In order to prove the case (b) of the current theorem we observe that the gate $W$ normalises $\Pi_f(\SU(d))$ (and therefore also $\LOF$). Moreover, a direct computation shows that $W^2 \propto \I$ and therefore $W^2\in \T(\Hfer)$. Thus any gate form $\langle \LOF, W \rangle$ can be presented as element form $\LOF \cup \LOF \cdot W$.

\end{proof}

\begin{customthm}{4}[Extensions of active Fermionic Linear Optics by an additional gate] \label{thm:FLOGATES2}
Let $\langle \FLO, V\rangle$ be the group of transformations generated by active linear optics and $V$ acting in positive-parity Fock subspace $\Hfree$ with $d>3$ modes. We have the following possibilities:
\begin{itemize}
\item[(a)] If $d\neq 2k$, then $\langle \FLO, V\rangle=\U\left(\Hfree\right)$ 
\item[(b)] If $d=2k$, and $[V\ot V, \L_{\FLO}]=0$, then  $
\langle \FLO, V\rangle = G_{\FLO}=\lbrace \left. U\in \U\left(\Hfree \right) \right|    [U\ot U, \L_{\FLO}]=0\ \rbrace$.
\item[(c)] If $d=2k$ and $[V\ot V, \L_{\FLO}]\neq 0$, then  $\langle \FLO, V\rangle=\U\left(\Hfree\right)$.
\end{itemize}
\end{customthm}

\begin{proof}
The proof follows essentially the same steps as the one of Theorem \ref{thm:BOSGATES2}. Similarly to  the bosonic case we have $\FLO=\langle \Pi^+_\FLO (\Spin(2d)),\T(\Hfree)\rangle$ and we can assume without the loss of generality that $V\in\SU(\Hfree)$.  Then, we have 
\begin{equation}\label{eq:convenientGENflo}
\langle \FLO, V \rangle = \langle \GEN , \T(\Hfree) \rangle\ ,\ \text{where }\ \GEN =\langle \Pi^+_\FLO (\Spin(2d)) , V \rangle\ .
\end{equation}
Because of Lemma \ref{lem:normFLO} we see that the normaliser of $\Pi^+_\FLO (\Spin(2d))$ consists of elements form $\T(\Hfree)$. The proofs for cases (a), (b), (c) are then exactly analogous to their bosonic counterparts.
\end{proof}

\end{document}